\documentclass[12pt]{article}
\pdfoutput=1

\usepackage{amsfonts}
\usepackage{amsmath,bm}
\usepackage{amssymb}
\usepackage{graphicx,psfrag,epsf,float}
\usepackage{placeins} 
\usepackage{multirow}
\usepackage{amsthm}
\usepackage{enumerate}
\usepackage{natbib}
\usepackage{booktabs}
\usepackage[hyphens]{url} 
\usepackage{xcolor}
\usepackage{soul}
\usepackage[normalem]{ulem}
\usepackage{hyperref}
\usepackage{authblk}
\usepackage{subcaption}
\usepackage{array}
\usepackage{rotating}
\usepackage{makecell}
\hypersetup{hidelinks}
\usepackage[bottom]{footmisc}
\newtheorem{theorem}{Theorem}
\newtheorem{lemma}{Lemma}
\theoremstyle{definition}
\newtheorem{definition}{Definition}
\newtheorem{assumption}{Assumption}

\usepackage{xr}
\newcommand{\blind}{1}

\definecolor{emerald}{RGB}{12, 166, 151}
\definecolor{lava}{rgb}{0.81, 0.06, 0.13}

\newcommand{\given}{\mathbin{|}}

\newcommand{\sumStatsTrain}{\Tilde{\bm{\beta}}^\mathrm{(sum)}_{\text{train}}}

\newcommand{\bsumStats}{\bm{\beta}_\mathrm{sum}}

\newcommand{\ld}{\bm{D}}
\newcommand{\ldRef}{\bm{D}_{\text{ref}}}
\newcommand{\VRef}{\bm{V}_{\text{ref}}}

\newcommand{\bv}{\bm{v}}
\newcommand{\eigenvalRef}{\bm{\Delta}_{\text{ref}}}

\newcommand{\bbeta}{\bm{\beta}}
\newcommand{\Ntrain}{N} 

\newcommand{\priorSD}{\lambda}
\newcommand{\priorSDs}{\lambda^2_1,...,\lambda^2_P}
\newcommand{\priorSDmin}{\lambda_{\text{min}}}
\newcommand{\priorSdBd}{\sigma_\textrm{bd}}
\newcommand{\eigenval}{\delta}
\newcommand{\bLambda}{\bm{\Lambda}}

\newcommand{\colsp}{\operatorname{col}}

\newcommand{\thinnerspace}{\mskip.5\thinmuskip}

\let\originalleft\left
\let\originalright\right
\renewcommand{\left}{\mathopen{}\mathclose\bgroup\originalleft}
\renewcommand{\right}{\aftergroup\egroup\originalright}

\begin{document}

	\def\spacingset#1{\renewcommand{\baselinestretch}%
		{#1}\small\normalsize} \spacingset{2}

	
	\if1\blind
	{
		\title{\bf Constructing Genetic Risk Scores: Robust Bayesian Approach through Projected Summary Statistics and Flexible Shrinkage}
		\author[1]{Yuzheng Dun}
		\author[1,3]{Nilanjan Chatterjee}
		\author[2,*]{Jin Jin}
		\author[1,*]{Akihiko Nishimura}
		\affil[1]{Department of Biostatistics, Bloomberg School of Public Health, Johns Hopkins University}
		\affil[2]{Department of Biostatistics, Epidemiology and Informatics, University of Pennsylvania}
		\affil[3]{Department of Oncology, School of Medicine, Johns Hopkins University}
		\affil[*]{Corresponding authors: Akihiko Nishimura, anishim2@jhu.edu;
			Jin Jin, jin.jin@pennmedicine.upenn.edu}
		\date{}
		\maketitle
	} \fi
	
	\if0\blind
	{
		\bigskip
		\bigskip
		\bigskip
		\begin{center}
			{\LARGE\bf Constructing Genetic Risk Scores: Robust Bayesian Approach through Projected Summary Statistics and Flexible Shrinkage}
		\end{center}
		\medskip
	} \fi

	\newpage
	\begin{abstract}
		Polygenic risk scores (PRS) developed from genome-wide association studies (GWAS) can be used for risk stratification by quantifying the genetic contribution to disease, and many clinical applications have been proposed. 
		Bayesian methods are popular for building PRS because of their natural ability to regularize models and incorporate external information.
		In this article, we present new theoretical results, methods, and extensive numerical studies to advance Bayesian methods for PRS applications. 
		We identify a potential risk, under a common Bayesian PRS framework, of posterior impropriety when integrating the required GWAS summary statistics and linkage disequilibrium (LD) data from distinct sources. 
		As a principled remedy, we propose a projection of the summary statistics that ensures compatibility between the two sources and in turn a proper behavior of the posterior. 
		We further introduce a new PRS method, with accompanying software, under the less-explored Bayesian bridge prior to more flexibly model varying sparsity levels in effect-size distributions. 
		We extensively benchmark it against alternative Bayesian methods using synthetic and real datasets, quantifying the impact of prior specification and LD estimation strategy. 
		Our proposed PRS-Bridge, equipped with the projection technique and flexible prior, demonstrates the most consistent and generally superior performance across a variety of scenarios.
	\end{abstract}
	
	\noindent%
	\vfill
	
	\newpage
	\spacingset{1.75} 
	
	\section{Introduction}
	\label{sec:intro}
	
	Genome-wide association studies (GWAS) have identified tens of thousands of inherited variants explaining a considerable amount of variations in human traits and diseases across individuals \citep{seibert2018polygenic, uffelmann2021genome}. 
	Polygenic risk scores (PRS), constructed as weighted sums of the number of risk alleles, summarize the effects of genetic markers for human traits and diseases.
	PRS are increasingly used as a quantitative summary of individuals’ genetic liability for complex diseases.
	Extensive research have demonstrated the utility of PRS in disease risk predictions \citep{chatterjee2016developing, khera2018genome}, population screening for early detection of cancers \citep{seibert2018polygenic,mchugh2025assessment,esserman2025risk}, refining biomarkers for diagnoses \citep{kachuri2023genetically}, improving efficiency of clinical trial designs \citep{liu2024integration, german2025incorporating}, and elucidating disease biology \citep{Pradeep2017Polygenic,smith2024multi,pascat2026partitioned}.
	Various clinical uses of polygenic scores are also reviewed by \citet{torkamani2018personal,kullo2022polygenic}, and \citet{kullo2025clinical}.
	
	At a basic level, the task of building PRS can be viewed as a high-dimensional regression problem for predicting outcome variables from hundreds of thousands or millions of genetic variants. 
	Bayesian methods \citep{Campos2010Predicting, vilhjalmsson2015modeling,zhu2017bayesian,lloyd2019improved,ge2019polygenic, ldpred2} have been popular in the PRS field as placing a prior on SNP effect sizes provides a natural way to introduce regularization and borrow external information \citep{ hu2017leveraging, marquez2021incorporating}.
	
	Despite its shared mathematical structure with standard regression problems, PRS development involves unique structures and challenges that complicate an application of existing Bayesian high-dimensional regression machinery. 
	One challenge is the difficulty of access to sizable individual-level data.
	Correspondingly, PRS development typically relies on more accessible GWAS summary statistics data
	obtained from one-SNP-at-time regression. 
	 Reconstructing a multivariate regression model from the summary statistics then requires an estimate, typically obtained from external reference data, of linkage disequilibrium (LD) to account for correlations among genetic variants \citep{vilhjalmsson2015modeling,mak2017polygenic,wang2020simple,pattee2020penalized,chen2021penalized}.
	 While this framework has been utilized extensively, statistical ramifications of integrating the two heterogeneous data sources in Bayesian PRS modeling have never been formally investigated.
	 In fact, the discrepancy in underlying data between GWAS summary statistics and an LD estimate, if not properly accounted for, leads to unexpected and potentially catastrophic inference as we establish in this article. 
	
	Another unique feature of PRS modeling is its unusual effect size distribution or, in a more statistical language, sparsity structure in regression coefficients.
	The high-dimensional regression literature typically assumes a relatively small number of the predictors explain most of the variation in the outcome. 
	This is often not the case in polygenic prediction, however, where thousands to tens of thousands of small-effect genetic variants cumulatively contribute to prediction.
	Further, complex traits demonstrate a wide variety of genetic architectures \citep{zhang2018estimation}, calling for a prior with flexibility to adapt to a wide range of potential effect size distributions.
	While a range of prior families have been explored in the PRS literature \citep{Campos2010Predicting, vilhjalmsson2015modeling,zhu2017bayesian, ge2019polygenic, ldpred2}, the impact of prior choice on model performance remains poorly understood since different methods have been benchmarked under different data sets and LD estimation strategies.

	In this article, we make several important contributions to the Bayesian PRS literature.
	First, we identify and quantify a serious pitfall in the use of the commonly used approximate likelihood based on GWAS summary statistics and genetic correlation matrix from two distinct data sources. 
	We show how the ill-defined nature of the approximate likelihood makes the inference susceptible to the prior choice or even produces an improper posterior. 
	We then develop a principled solution to this ill-behavior, which in the past has been dealt with in an ad hoc manner such as by constraining the prior variance of regression coefficients \citep{ge2019polygenic}.
	Second, we demonstrate the importance of using a prior with sufficient flexibility to account for different genetic architectures. 
	Concretely, we introduce PRS-Bridge, a new PRS method based on the bridge prior \citep{polson2014bayesian}.
	Compared to more popular shrinkage priors such as the horseshoe prior \citep{carvalho2009handling}, the bridge prior better adapts to varying sparsity levels and effect size distributions through its exponent parameter. 
	While there are other shrinkage priors providing similar flexibility \citep{brown2010inference}, we deploy the bridge prior also to take advantage of its major computational advantage in allowing collapsed Gibbs sampling in posterior inference \citep{polson2014bayesian,aki2022prior}.
	Finally, we provide reliable evidence on the relative performance of different Bayesian PRS methods by carrying out one of the most systematic and comprehensive benchmark studies to date.
	We use both simulated and real datasets to benchmark PRS-Bridge against three representative PRS methods---LDpred2,  PRS-CS, and Lassosum---across a range of genetic architectures, data sources, and LD estimation strategies.
	Our benchmark study demonstrates PRS-Bridge's compelling performance and a significant impact of LD estimation strategy on PRS methods' performances.
	
	
	An optimized, open-source implementation of the algorithm is available from a GitHub repository at \url{https://github.com/YuzhengDun1999/PRSBridge}.

	\section{Bayesian PRS Modeling Based on GWAS Summary Statistics and External LD Reference Data}\label{likelihood}
	
	In this section, we lay out the general statistical framework underlying the development of PRS based on GWAS summary data and external LD reference data. 
	We will explain and quantify the statistical and computational ill-behavior that arises from the use of two separate data sources for summary statistics and LD estimation.
	
	\subsection{Real-World Motivation: Developing PRS When Individual-Level Data are Scarce}
	Before proceeding to present the statistical framework, we first introduce a motivating real-data application---predicting individuals' genetic risk for specific diseases---on which we demonstrate our method in Section~\ref{Numerical_studies}.

	For most diseases, we rarely have a sufficient number of cases with individual-level genotype data available to train a predictive model directly.
	In many settings, therefore, PRS model building relies almost exclusively on combining publicly available GWAS summary statistics and an external LD reference panel, where the former provides marginal SNP effect sizes and the latter provides information on the correlation structure among SNPs.
	Table~\ref{tab:GWAS_disease} below provides examples of publicly available GWAS datasets on five disease traits, which we will combine with reference panels from the UK biobank and the 1000 Genomes Project in our benchmark study of PRS methods (Section~\ref{Numerical_studies}).
	Note how the summary-level datasets provide far larger sample sizes than the individual-level data provided by the UK biobank data. 
	This difference in sample size highlights the importance of GWAS summary data-based PRS methods.
	
	Despite their widespread use, PRS construction from summary statistics faces a critical, yet under-investigated, challenge.
	The required GWAS summary statistics and LD reference data are almost always obtained from two separate studies with distinct cohorts. 
	This raises the question: when the two datasets arise from the same underlying genetic population, how can we perform valid inference for SNP effect sizes and generate reliable risk prediction?
	Most existing methods implicitly assume this data mismatch to be ignorable;
	however, PRS modeling in practice often encounters unstable behaviors such as non-convergence of effect estimates. 
	The key issue is that, when the two data sources are incompatible, the commonly used likelihood formulation for summary statistics does not necessarily represent a coherent probabilistic model.
	We elucidate this issue and develop a principled solution in the subsequent sections.
	
		\begin{table}[htb]
			\centering
			\caption{%
				Information on the GWAS summary, tuning, and validation data for the five disease traits studied in Section~\ref{Real_data_binary}.
				Incident cases and controls from the UK Biobank are evenly split into tuning and validation sets.
			}
			\label{tab:GWAS_disease}
			\small
			\renewcommand{\arraystretch}{1.5} 
			\setlength{\tabcolsep}{0.5pt} 
			\spacingset{1}
			\begin{tabular}{ >{\centering\arraybackslash}m{3.4cm}>{\centering\arraybackslash}m{3.4cm}>{\centering\arraybackslash}m{4cm}>{\centering\arraybackslash}m{5.3cm}}
				\toprule
				Disease & GWAS Reference & GWAS Sample Size\linebreak (Case/Control)& \mbox{Tuning and Validation}\linebreak \mbox{Sample Size (Case/Control)}\\ 
				\midrule
				Breast Cancer & \citet{michailidou2017association} & 228,951\linebreak (122,977/105,974) & 8,340\linebreak (4,170/4,170)\\
				Coronary Artery Disease & \citet{cardiogramplusc4d2015comprehensive} & 184,305\linebreak (60,801/123,504) & 33,016\linebreak (14,054/20,000)\\ 
				Depression & \citet{wray2018genome} & 124,430\linebreak (45,645/97,674) & 19,355\linebreak (6,383/20,000)\\ 
				Inflammatory Bowel Disease & \citet{liu2015association} & 34,652\linebreak (12,882/21,770) & 10,332 \linebreak(2,966/20,000)\\ 
				Rheumatoid Arthritis & \citet{okada2014genetics} & 58,284\linebreak (14,361/43,923) & 14,053\linebreak (4,262/10,000)\\ 
				\bottomrule
			\end{tabular}
		\end{table}
	
	\subsection{Likelihood Based on GWAS Summary Statistics and External LD Reference Data}\label{likehood_sumstat}
	When individual-level data is available, PRS model development can be viewed as a linear regression problem with the familiar mathematical structure:
	\begin{equation}\label{lm_model}
		\bm{y}=\bm{X}\bm{\beta}+\bm{\epsilon}, \ \bm{\epsilon} \sim \mathcal{N}\left(\bm{0},\sigma_{\epsilon}^2\bm{I}\right),
	\end{equation}
	where $\bm{y}$ is the phenotype vector of $N$ training individuals, $\bm{X}$ the $N\times P$ genotype matrix, $\bm{\beta}$ the $P$-dimensional vector representing effect sizes of underlying genetic variants, and $\bm{\epsilon}$ the error term.
	The demographic covariates such as age, gender, and the top ten principal components are usually regressed out first so that the PRS model only includes the genetic markers as predictors \citep{choi2020tutorial}.
	A set of genetic variants typically included in the PRS model comprises single nucleotide polymorphisms (SNPs), the most common type of DNA variations across individuals. 
	The genotype of a SNP takes values of 0, 1, or 2 depending on the number of the risk allele copies carried by individuals in their two parental chromosomes. 
	For the rest of our discussion in this section, we assume both $\bm{y}$ and $\bm{X}$ have been centered and standardized.
	
	While conceptually equivalent to a linear regression problem, PRS development is distinguished by its routine reliance on GWAS summary-level data instead of less accessible individual-level data.
	From standardized individual-level data, GWAS summary data is generated as estimates of marginal effect sizes through one-SNP-at-a-time regression:
	$$\beta_{\mathrm{sum},j}=\bm{x}_j^T\bm{y}/N \ \text{ for } \, j=1,\ldots,P,$$ 
	where $\bm{x}_j$ denotes the $j$-th column of the standardized genotype matrix $\bm{X}$. 
	If we had access to the individual-level data $\bm{X}$, the model (\ref{lm_model}) for individual-level data implies that the likelihood of summary statistics $\bsumStats = N^{-1} \bm{X}^T \bm{y}$ would be given by
	\begin{equation}
	\label{eq:indiv_level_summary_data_likelihood}
	\bsumStats\given \bm{X},\sigma_{\epsilon}^2\sim\mathcal{N} \left(\frac{\bm{X}^T\bm{X}}{N}\bm{\beta}, \frac{1-h^2}{N}\frac{\bm{X}^T\bm{X}}{N} \right),
	\end{equation}
	where $h^2=1-\sigma_{\epsilon}^2$ represents the amount of trait variation attributable to the additive effect of genetic variants and is referred to as heritability in population genetics.
	
	We cannot, in practice, use the likelihood \eqref{eq:indiv_level_summary_data_likelihood} to infer the effect size $\bbeta$ since the summary statistics comes without the individual-level genotype matrix $\bm{X}$ and thus without the empirical covariance $\bm{X}^T \bm{X}$.
	One way to get around this problem is to view the rows of $\bm{X}$ as independent and identically distributed draws and take the expectation over this underlying population.
	Following the derivations in \cite{zhu2017bayesian}, we can then derive an approximate likelihood in terms of the population correlation $\bm{D} = \mathop{\mathbb{E}}\left(\bm{X}^T\bm{X}/N\right)$:
	\begin{equation}\label{likelihood_sum0}
	\bsumStats \given \bm{\beta}, \bm{D}\sim\mathcal{N}\left( \bm{D}\bm{\beta}, \frac{\bm{D}}{N} \right),
	\end{equation}
	
	\noindent
	where $\bm{D}$ is typically referred to as the LD matrix in population genetics.
	We then substitute the population-level $\bm{D}$ with an empirical estimate based on an external reference panel.
	For this purpose, we can use publicly available individual-level data, such as those through the popular 1000 Genomes (1000G) Project \citep{siva20081000}.
	By substituting this estimated LD matrix $\bm{D}_{\text{ref}} = \bm{X}_\textrm{ref}^T \bm{X}_\textrm{ref}/N_\textrm{ref}$ for $\bm{D}$ in the likelihood (\ref{likelihood_sum0}), we obtain the following approximate likelihood based on GWAS summary data and external LD reference data:
	\begin{equation}\label{likelihood_sum}
		\bsumStats\given \bm{\beta},\bm{D}_{\text{ref}} \sim \mathcal{N}\left( \bm{D}_{\text{ref}}\bm{\beta},\frac{\bm{D}_{\text{ref}}}{\Ntrain} \right).
	\end{equation}
	
	Many of the existing PRS methods are based on the above approximate likelihood for summary statistics.
	As we will show in next section, however, this approximate likelihood is not a proper likelihood.
	Importantly, its naive use under the Bayesian framework results in a posterior distribution with unexpected and undesirable behavior.

	\subsection{Nominal Posterior's Ill-behavior Caused by Mismatch between GWAS Summary Statistics and LD Reference Data}\label{ill-behavior}
	In Section \ref{likehood_sumstat}, we presented an approximate likelihood, widely used in the PRS problem, that combines summary data from one data source and LD matrix from another.
	We now show that the approximation can break down under a mismatch between the two data sources, leading to an unexpected and potentially catastrophic inference.
	
	The issue stems from the fact that an empirical estimate $\ldRef$ of a high-dimensional covariance matrix is unstable and, worse, often rank-deficient due to the limited reference sample size and high correlation among SNPs.
	When $\bm{D}_{\text{ref}}$ is singular, the likelihood (\ref{likelihood_sum}) represents a degenerate Gaussian distribution supported on the subspace $\text{col}(\ldRef) \subsetneq \mathbb{R}^P$, where $\text{col}(\ldRef)$ denotes the column space of $\ldRef$.
	This degeneracy is not a problem on its own;
	in fact, when $\ldRef$ is estimated from the same data source as $\bsumStats$ and hence $\bm{X}_\textrm{ref} = \bm{X}_\textrm{train}$, the likelihood makes sense because $\bsumStats = N_\textrm{ref}^{-1} \bm{X}_\textrm{ref}^T \bm{y}_\textrm{ref}$ indeed lies in the column space of $\ldRef = N_\textrm{ref}^{-1}\bm{X}_\textrm{ref}^T \bm{X}_\textrm{ref}$.
	For $\bsumStats$ and $\ldRef$ coming from two different data sources, however, it is likely to have a situation $\bsumStats \notin \text{col}(\ldRef)$;
	 i.e.\ $\bsumStats$ lying outside the support of the degenerate Gaussian likelihood \eqref{likelihood_sum}.
	This represents an impossible event and, when this occurs, the approximate likelihood for $\bsumStats$ makes no sense as a function of the parameter $\bm{\beta}$.

	Curiously, if we were to ignore the ill-defined nature of the approximate likelihood, we can formally derive an apparently proper posterior distribution under a Gaussian prior $\bm{\beta} \sim \mathcal{N}\left(\bm{0}, \bm{\Sigma}_0\right)$.
	This posterior, which we call a \textit{nominal} posterior, is given as follows:
	\begin{equation}\label{nominal_post}
		\bbeta\given\bsumStats,\ldRef,\bm{\Sigma}_0\sim\mathcal{N}\left(\left(\Ntrain \ldRef+\bm{\Sigma}_0^{-1}\right)^{-1}\Ntrain\bsumStats,\left(\Ntrain\ldRef+\bm{\Sigma}_0^{-1}\right)^{-1}\right).
	\end{equation}
	The apparent propriety of this nominal posterior likely explains why the ill-defined likelihood has so far been overlooked in the PRS literature.
	As we will show in Theorem~\ref{thm2} below and additionally in Theorem~\ref{thm1} of the supplement Section~\ref{proof1}, however, the nominal posterior exhibits highly problematic behaviors.

	To deal with high-dimensionality of the problem, Bayesian PRS methods typically deploy a sparsity-inducing prior on $\bm{\beta}$ with a scale-mixture representation $\beta_j \given \tau, \priorSD_j \sim \mathcal{N}(0, \tau^2 \priorSD_j^2)$ with global scale $\tau$, which we here assume to be fixed for simplicity, and local scale  $\priorSD_j\sim\pi(\priorSD_j)$.
	The posterior inference based on these methods typically relies on a Gibbs sampler, alternately sampling from the nominal posterior of \eqref{nominal_post} with $\bm{\Sigma}_0 = \tau^2 \bm{\Lambda}^2$ for $\bm{\Lambda} = \operatorname{diag}(\lambda_1, \ldots, \lambda_P)$ and from $\bm{\lambda} \given \bm{\beta}, \tau$. 
	This Gibbs sampler may never converge, however; two proper conditional distributions do not necessarily form a proper joint distribution  \citep{Andrew1993characterizing, hobert1998functional_compatibility}.
	This is a particularly germane concern given the ill-defined nature of the approximate likelihood under the data mismatch.
	
	The joint posterior distribution, if it exists, can be expressed in terms of the conditional distributions as
	\begin{equation}\label{joint_post}
		\pi\left(\bbeta, \bm{\lambda} \given \bsumStats,\ldRef, \tau \right)
		\propto
			\frac{
				\pi\left(
					\bbeta \given \bm{\lambda}, \bsumStats,\ldRef, \tau
				\right)
				\, \pi\left(
					\bm{\lambda} \given\hat{\bbeta},\bsumStats,\ldRef, \tau
				\right)
			}{
				\pi\left(\hat{\bbeta} \, \middle| \,  \bm{\Sigma}_0,\bsumStats,\ldRef,  \tau\right)
			}
	\end{equation}
	for any $\hat{\bbeta}$ \citep{Andrew1993characterizing}.
	We will refer to the potentially improper density function satisfying the above relation as the \textit{joint nominal posterior}.
	We show below that, if the prior on $\beta_j$ has a tail that decays slower than exponential, then the joint nominal posterior is improper under the mismatched case.
	Importantly, while each of the two conditional distributions is well-defined on its own, the Gibbs sampler alternately sampling from each never converges.
	
	In the theorem statement below, we denote the sample covariance of the GWAS training data by $\ld_{\mathrm{train}} = \Ntrain^{-1} \bm{X}_\mathrm{train}^T \bm{X}_\mathrm{train}$.
	The theorem technically requires an assumption (Assumption~\ref{continuous}) to preclude a situation that is practically impossible but mathematically imaginable; 
	we defer its statement to the supplement due to its limited practical interest.
	
	\begin{theorem}\label{thm2}
		The following result holds for almost every realization of $\bsumStats$ under Assumption~\ref{continuous}. 
		Consider a joint nominal posterior distribution $\bm{\beta},\priorSDs \given \bsumStats, \bm{D}_{\mathrm{ref}}, \tau$ under a heavy-tailed prior on $\beta_j$, represented as a scale-mixture  $\beta_j \given \tau, \priorSD_j \sim \mathcal{N}(0, \tau^2 \priorSD_j^2)$ with $\priorSD_j\sim\pi(\priorSD_j)$.
		The joint nominal posterior is improper when $\operatorname{null}( \bm{D}_{\mathrm{ref}} ) \not\subset \operatorname{null}(\bm{D}_{\mathrm{train}})$,
		but proper when $ \bm{D}_{\mathrm{ref}} = \bm{D}_{\mathrm{train}} $ or, more generally, when $\operatorname{null}( \bm{D}_{\mathrm{ref}} ) \subset \operatorname{null}(\bm{D}_{\mathrm{train}})$.
	\end{theorem}
	\noindent
	The proof, along with a formal definition of heavy-tailed distributions, is in the supplement Section~\ref{proof2}.
	
		Note that the mismatched case, and the resulting impropriety of the nominal posterior, is more the norm than the exception in practice.
		The reference sample size used to calculate the $P \times P$ empirical covariance $\bm{D}_{\mathrm{ref}}$ is typically much smaller than both the number of SNPs and the training sample size, making the null space of $\bm{D}_{\mathrm{ref}}$ larger than that of $\bm{D}_{\mathrm{train}}$ in general and thereby causing the mismatch condition $\operatorname{null}( \bm{D}_{\mathrm{ref}} ) \not\subset \operatorname{null}(\bm{D}_{\mathrm{train}})$.
		For practical application of Bayesian PRS methods, it is thus essential to address this mismatch problem, whether through the projection approach we propose below or through ad hoc alternatives.
		We expect degrees of mismatch between the reference and training data, such as the angle between $\bsumStats$ and $\operatorname{col}( \bm{D}_{\mathrm{ref}} )$, to affect relative effectiveness of these mismatch remedies and the resulting PRS methods' practical performances.
		Further exploration of these issues, beyond our theoretical result above, is a worthy topic for future work.
	
	\subsection{Real Data Demonstrations of Danger from Data Mismatch}
	\label{sec:ill-behavior_demo}
	
	In Section~\ref{ill-behavior}, we have theoretically characterized the ill-behavior arising from the use of the approximate likelihood in the presence of a mismatch between LD reference and GWAS training samples.
	We now demonstrate, using a real dataset and existing Bayesian PRS software, that this issue has a real consequence in application.

	Our demonstration uses the PRS-CS software by \citet{ge2019polygenic}, the model behind which we discuss in more detail in Section~\ref{existingmethods} below.
	We obtain GWAS summary statistics for BMI from approximately 290K unrelated European individuals in UK Biobank, and the LD matrix estimated from 503 European individuals in 1000G from the PRS-CS website.
	We then run PRS-CS on this dataset, using the software as is except for a one-line code modification to remove their ad hoc fix for the data mismatch issue, as we explain below.
		
	The prior used in PRS-CS is heavy-tailed and, by Theorem~\ref{thm2}, yields an improper posterior under the data mismatch.
	The resulting breakdown of posterior inference, though \textit{not} its root cause, has likely been known to the PRS-CS authors as they propose to impose a constraint $\tau^2 \lambda_j^2 \leq \priorSdBd^2$ on the prior variance of $\beta_j$ for some fixed $\priorSdBd$ \citep{ge2019polygenic}.
	They explain this constraint as necessary to deal with co-linearity among SNPs, but this logic is faulty:
	while problematic in its own way, co-linearity will not cause an improper posterior under a proper prior.
	Figure~\ref{fig:PRScs_traceplot_new} shows traceplots of the regression coefficients when the constraint is removed.
	Since the Gibbs sampler targets an improper posterior, the coefficients start exploding at some point, with values exceeding $10^{20}$ in magnitudes.
	When run longer, the software breaks down before long due to numerical errors.
	
	While the constraint $\tau^2 \lambda_j^2 \leq \priorSdBd^2$ does ensure posterior propriety, this ad hoc solution only indirectly addresses the fundamental problem of the ill-defined approximate likelihood \eqref{likelihood_sum}.
	One practical concern is the posterior inference's sensitivity to the choice of the bound $\priorSdBd$ since, as our theory tells us, this constraint is the only thing preventing the posterior from becoming improper.
	When using this thresholding approach, therefore, $\priorSdBd$ should really be considered as a potential tuning parameter, rather than assuming $\priorSdBd = 1/\sqrt{\Ntrain}$ proposed casually as default by \citet{ge2019polygenic} to work well in all applications.
	
	\begin{figure}[htb]
		\centering
			\includegraphics[width=1\textwidth]{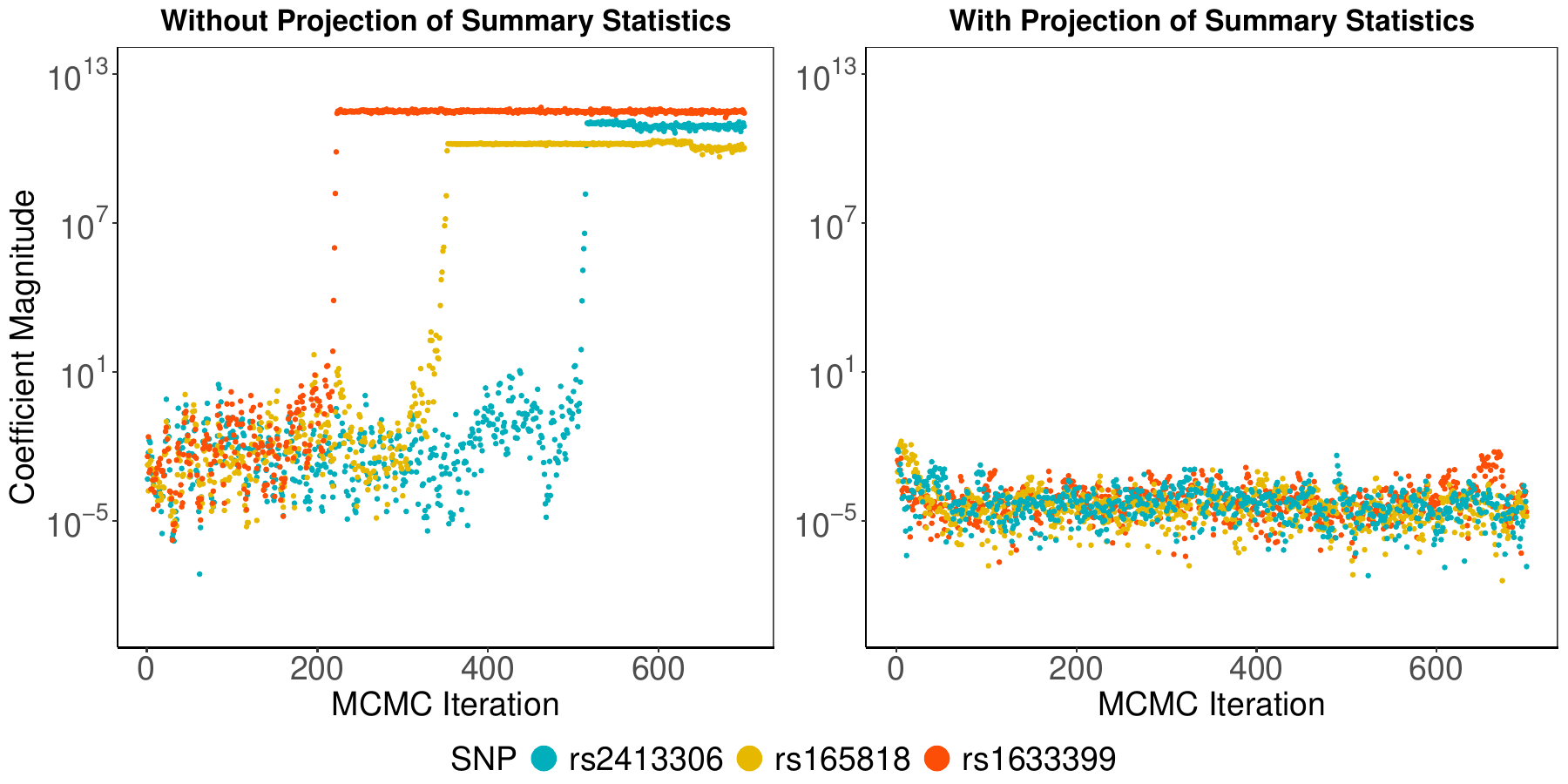}
			
		\caption{%
			Traceplots of the samples, for the first three coefficients to explode, generated by PRS-CS when removing the ad hoc constraint on the prior variance.
			The summary statistics is taken from UK Biobank and the LD matrix from 1000G.
			The data mismatch results in impropriety of the joint nominal posterior and in the explosion of the Gibbs sampler. 
			The software breaks down after a while due to numerical errors.
			The use of projected summary statistics ensures proper posterior inference.
		}
		\label{fig:PRScs_traceplot_new}
	\end{figure}

	\subsection{Projected GWAS Summary Statistics}
	\label{lowRankApprox}
	Having demonstrated the posterior ill-behavior caused by data mismatch, we turn to providing a statistically principled solution to the problem.

	The discussion and theoretical result of Section~\ref{ill-behavior} characterize the ill-behavior as caused by the summary statistics lying outside the column space of the reference LD matrix. 
	This insight motivates the following approach based on a linear projection of summary statistics.
	The idea is to treat the approximate likelihood \eqref{likelihood_sum} as a likelihood not for the raw summary statistics $\bsumStats$ but for its projection $\bm{P}_\textrm{ref} \, \bsumStats$ onto $\colsp(\ldRef)$.
	This projection can be expressed as
	$ \bm{P}_\textrm{ref} \, \bsumStats = \sum_{k=1}^{K} \langle \bsumStats, \bv_k \rangle \bv_k $
	in terms of the eigenvectors $\bv_1,\bv_2,...,\bv_K$ associated with the non-zero eigenvalues of $\ldRef$.
	The projected summary statistics is guaranteed to lie in the support of the approximate likelihood \eqref{likelihood_sum}; its use in place of the raw one, therefore, ensures a proper inference. 
	
	This projection approach can be applied to any PRS methods to avoid the posterior ill-behavior.
	For illustration, we modify PRS-CS to use the projected summary statistics.
	This modified version yields a sound posterior inference without the constraint on the prior variance, as evidenced in the traceplots on the right panel of Figure~\ref{fig:PRScs_traceplot_new}, since the posterior is now guaranteed to be proper.
	
	Our projection approach not only resolves the theoretical issue in deploying the approximate likelihood but also delivers a compelling performance in practice as demonstrated in the numerical study of Section~\ref{Numerical_studies}.
	On the other hand, it is possible that the part of summary statistics discarded by the projection---i.e.\ the component $(\bm{I} - \bm{P}_\textrm{ref}) \, \bsumStats$ orthogonal to $\colsp(\ldRef)$---contain potentially useful information on predicting the genetic risk for the target population, even if it is incompatible with the approximate likelihood \eqref{likelihood_sum} based on the reference LD.
	One possibility, therefore, is to utilize this orthogonal component by specifying its relation to the PRS model coefficients through an auxiliary likelihood function, in addition to the approximate likelihood for $\bm{P}_\textrm{ref} \, \bsumStats$.
	As it turns out, this auxiliary modeling of $(\bm{I} - \bm{P}_\textrm{ref}) \, \bsumStats$ provides a likelihood-based interpretation of ad hoc fixes used by existing PRS methods.
	In particular, it can be shown that 
	1) a special case of the auxiliary modeling yields the same effect on the posterior inference as replacing $\ldRef$ with a regularized version and 2) PRS-CS's constraint  on the prior variance essentially amounts to replacing $\ldRef$ with $\ldRef + \bm{I}$.
	While a more thorough investigation of this auxiliary modeling possibility is left for future work, we explain its connection to the existing methods in Section~\ref{sec:model_null_space} and present some preliminary results in Section~\ref{sec:add_num_result} of the supplement.

	\section{Methods}\label{methods}
	\subsection{
		Benchmark PRS methods: LDpred2, PRS-CS, and Lassosum
	}
	\label{existingmethods}
	Here we review the three representative PRS methods---LDpred2, PRS-CS, and Lassosum---against which we benchmark our PRS-Bridge in Section~\ref{Numerical_studies}.
	LDpred2 and PRS-CS are popular Bayesian methods and thus provide the more direct methodological parallels for evaluating PRS-Bridge. 
	Lassosum is non-Bayesian but is one of the most popular PRS methods and serves as a useful point of reference in assessing our method's performance.
	
	LDpred \citep{vilhjalmsson2015modeling} models the effect size distribution of SNPs with a spike-and-slab prior in the form,
	\begin{equation}
		\beta_j\sim\left\{
		\begin{aligned}
			&\mathcal{N}\left(0,\frac{h^2}{\varpi P}\right) && \text{with probability } \varpi, \\
			&0 &&\text{otherwise},
		\end{aligned}
		\right.
	\end{equation}
	where $P$ denotes the total number of SNPs, $\varpi$ the causal SNP proportion, and $h^2$ the total heritability. 
	Its successor LDpred2 adopts the same model but provides a faster and more robust implementation \citep{prive2018efficient}.
        LDpred2-grid estimates them based on the predictive performance of the corresponding PRS model on a tuning dataset, while
	LDpred2-auto estimates them in a fully Bayesian manner.
	
	
	
	PRS-CS \citep{ge2019polygenic} is another popular Bayesian PRS method. 
	It adopts a global-local shrinkage prior known as the Strawderman-Berger prior:
	\begin{equation}
	\label{eq:strawderman_berger}
	\beta_j\given\tau,\lambda_j\sim \mathcal{N}\left(0,\tau^2\lambda_j^2\right)
	\, \text{ for } \,
	\pi(\lambda_j) = \frac{\lambda_j}{(1+\lambda_j^2)^{\frac{3}{2}}}
	\text{ on } \lambda_j \geq 0,
	\end{equation}
	where global scale $\tau$ controls the overall sparsity level and local scales $\lambda_j$ allow variations in effect size magnitudes.
	The PRS-CS software provides an option to estimates $\Ntrain \tau$ in a fully Bayesian manner with a half-Cauchy prior. 
	The global scale is otherwise treated as a tuning parameter and estimated based on the prediction performance on a tuning dataset.
	
	
	Lassosum \citep{mak2017polygenic} is one of the most well-known among non-Bayesian PRS methods and applies the $\ell_1$-regularized lasso regression to the summary statistics likelihood.
	It additionally uses a regularized version of the LD matrix defined as 
	\begin{equation}
	\label{eq:regularized_ld}
		(1-\zeta) \ldRef + \zeta \bm{I}
	\end{equation}
	where $\zeta \in [0, 1]$ is a tuning parameter;
	the regularized LD matrix has a full rank and avoids the ill-behavior from the data mismatch (Section~\ref{ill-behavior}).
	Both the lasso penalty parameter and $\zeta$ are selected using a tuning dataset.
	Incidentally, the LD matrix regularization of Equation~\eqref{eq:regularized_ld} coincides with one specific way of incorporating the information from $(\bm{I} - \bm{P}_\textrm{ref}) \, \bsumStats$ into the likelihood (Section~\ref{lowRankApprox}).

	\subsection{PRS-Bridge: Robust, Flexible, and Scalable PRS Method}
	\label{PRSBridge}
	
	We now introduce our PRS-Bridge, a new method with accompanying software that combines the summary statistics projection approach of Section~\ref{lowRankApprox} with the flexibility of the bridge prior having a powered exponential distribution of the form
	$$ \beta_j\given\tau\propto \tau^{-1}\exp\left( - \left| \frac{\beta_j}{\tau} \right|^{\alpha} \right ) 
		\ \text{ for } \, \alpha > 0. $$
	Importantly, the exponent parameter $\alpha$ gives the bridge prior an ability to adapt to different degrees of sparsity in the effect sizes and hence to different genetic architectures.
	This flexibility is a major advantage over the other priors previously considered in the PRS literature.
	We treat $\alpha$ as a tuning parameter and select an optimal value based on the performance on a tuning dataset. 
	We find the candidate values of $\alpha \in \{0.125,0.25,0.5\}$ to provide sufficient flexibility for the PRS application.  
	To accommodate situations in which validation data of sufficient size is not available, we also develop a strategy to auto-tune $\alpha$ through an empirical Bayes approach, by maximizing the marginal likelihood through a stochastic approximation \citep{delyon1999stochastic_approx}.
	We provide further details on this auto-tuning version, as well as comparison of its performance against the validation-data-tuning version, in the supplement Section~\ref{PRS-Bridge-auto}.
	
	Figure~\ref{fig:prior} illustrates with contour plots the flexibility offered by the bridge prior's exponent parameter.
	When $\alpha = 1$, the prior coincides with the Laplace distribution and the model recovers the Bayesian lasso \citep{park2008bayesian_lasso}. 
	As $\alpha \rightarrow 0$, the prior becomes more peaked at 0 and heavier-tailed at the same time, inducing an increasingly sparser structure in the SNP effect size distribution.
	
		\begin{figure}[htb]
		\centering
		\begin{subfigure}[b]{0.44\textwidth}
			\includegraphics[width=\textwidth]{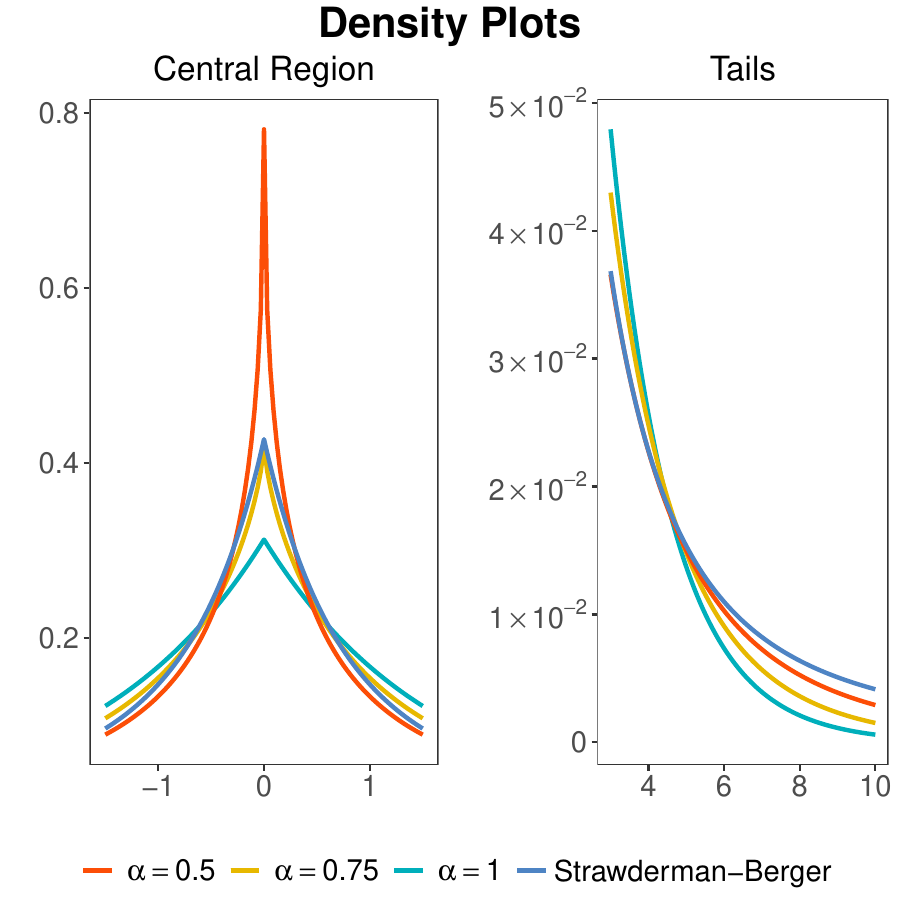}
			\label{fig:density_plot}
		\end{subfigure}
		~
		\begin{subfigure}[b]{0.44\textwidth}
			\includegraphics[width=\textwidth]{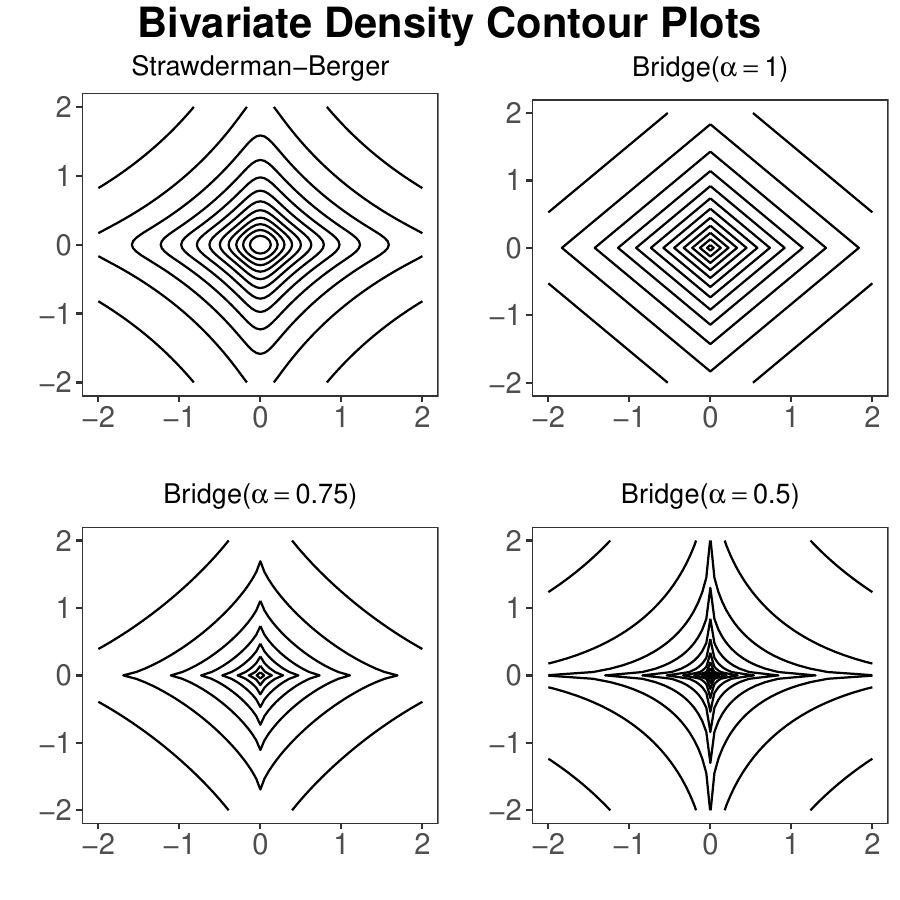}
			\label{fig:contour_plot}
		\end{subfigure}
		\vspace*{-\baselineskip}
		\caption{%
			Comparison of the Strawderman-Berger prior used in PRS-CS and the bridge prior under varying $\alpha$ values.
			In the left plot, the different priors are scaled to have the same amount of probability in the region $[-2, 2]$ to facilitate the comparison.
		}
		\label{fig:prior}
	\end{figure}
	
	
	The bridge prior also has a major computational advantage over other continuous shrinkage priors due to the availability of a collapsed update of the global scale parameter $\tau$ within the Gibbs sampler (\cite{polson2014bayesian}). 
	\cite{polson2014bayesian} show that this collapsed Gibbs sampler generates samples with little autocorrelation even when the posterior computation under other shrinkage priors suffers from poor mixing.
	This efficiency gain is particularly valuable given the high computational demands of the PRS application.
		
	For further scalability, our implementation of the PRS-Bridge method deploys the conjugate gradient (CG) sampler of \citet{aki2022prior} within the collapsed Gibbs sampler.
	The CG sampler speeds up the computation by transforming the task of sampling from a high-dimensional multivariate Gaussian into that of solving a deterministic linear system, to which the CG algorithm can be applied.
	For the PRS posterior, the main cost of the CG step is the iterative application of the matrix-vector operation $\bm{v} \to \ld \bm{v}$ by the LD matrix.
	This allows us to exploit the structured LD approximations discussed in Section~\ref{sec:snp_choice_and_ld_approx}.
	Further details on the posterior computation under PRS-Bridge is provided in the supplement Section~\ref{sec:CG_descrip}.
	
	
	\subsection{LD Approximation Strategy}
	\label{sec:snp_choice_and_ld_approx}
	Population genetic theory predicts that SNPs far apart in the genome are likely to be independent due to recombinations.
	An empirical estimate of SNP-correlation matrix, also referred to LD matrix, shows expected patterns of sparsity (Supplement Figure~\ref{ld}).
	It is common for PRS methods, therefore, to impose a sparsity structure in approximating the population LD matrix and try taking advantage of it in the computation. 
	Existing methods have explored various sparsity structures, such as banded and block diagonal approximations \citep{vilhjalmsson2015modeling,mak2017polygenic,ge2019polygenic,ldpred2}. 
	The block diagonal approximation, for example, allows each corresponding block of regression coefficients to be updated independently.
	
	The projection of the summary statistics proposed in Section~\ref{lowRankApprox} suggests an approach to approximate the LD matrix through a low-rank structure, which can be combined with the block diagonal approximation.
	Specifically, we selectively keep $K^*\leq K$ largest eigenvalues of LD matrix and disregard relatively smaller eigenvalues.
	The resulting low-rank LD matrix is given as $\ldRef=\sum_{k=1}^{K^*}\lambda_k\bv_k\bv_k^T$, where $\lambda_k$ is the eigenvalues of LD matrix corresponding to eigenvector $\bv_k$.
	This LD approximation is combined with the projection of summary statistics onto the space spanned by the eigenvectors $\bv_1, \ldots, \bv_{K^*}$.
	This low-rank approximation can in particular speed up the posterior computation based on the CG sampler.
	In our implementation of PRS-Bridge, we treat the percentage of eigenvectors projected away as a tuning parameter and find that low-rank approximation improves not only computational but also statistical efficiency.

	The impacts of LD approximation strategies on PRS methods' statistical performance have been understudied. 
	In our numerical study, we investigate impacts of using different LD approximation strategies in each PRS methods.
	 PRS-CS implements a block approximation of the LD matrix by partitioning the whole genome into nearly independent LD blocks of fixed sizes. 
	 For PRS-Bridge, we implement the same block approximation strategies but with additional flexibility to use larger LD block sizes. 
	 LDpred2 offers an option to use either the larger block approximation or banded approximation that only accounts for correlation within a neighboring LD region of each SNP. 
	Details of the choice of LD reference data and sparse LD structures are provided in the supplement Section~\ref{sec:LD_approx}.

	\section{Numerical Studies} \label{Numerical_studies}
	
	We first present a synthetic data simulation study in Section~\ref{simulation}, focusing on the three Bayesian PRS methods to assess how the prior choice affects the methods' performances under varying genetic architectures.
	Section~\ref{Real_data_continuous} and \ref{Real_data_binary} present real data studies with continuous and binary traits respectively, where we additionally include Lassosum in the comparison of predictive performances.
	Section~\ref{sec:additional_num_results} describes the additional benchmark studies, conducted to investigate precise sources of the observed performance differences, with their results presented in the supplement Section~\ref{sec:add_num_result}.

	
	We compare the methods' out-of-sample predictive performance as measured by the coefficient of determination $R^2$ for the continuous traits and by the transformed AUC for the binary traits. 
	The latter is defined as $2\{\Phi^{-1}(\textsc{auc})\}^2$, where $\Phi$ denotes the standard Gaussian cumulative distribution function, and allows us to directly translate the observed improvement in predictive accuracy into concrete gain in the form of reduced sample size requirement \citep{pharoah2002polygenic, chatterjee2013projecting}.
	For instance, a 10\% improvement in the transformed AUC means the method requires 10\% fewer samples to achieve the same predictive accuracy.
	The reduced sample size requirement has major implications for rare diseases and minority populations on which data may be limited.

	We quantify uncertainty in the predictive performance by repeating the out-of-sample prediction task under 100 different random splits of the individual-level data into tuning and validation datasets;
	we take the summary-level data used for the model training as fixed, following the standard practice \citep{ge2019polygenic}, since summary-level data in practice cannot be resampled.  
	We report the average performance over the 100 replications, along with $\pm 1.96$ times the standard deviation as an estimate of uncertainty.

	\subsection{Plasmode Synthetic Data from Spike-and-slab Model}
	\label{simulation}
	We first compare performances of LDpred2, PRS-CS, and PRS-Bridge on the synthetic large-scale datasets of \cite{Zhang2023new};
	these are ``plasmode'' datasets \citep{gadbury2008plasmode}, mimicking essential features of real data such LD structures and SNP proportions to allow for realistic evaluation.
	The genotype data of 120,000 individuals are simulated, using HAPGEN2 version 2.1.2 \citep{su2011hapgen2}, to mimic the 1000G reference data containing 498 unrelated individuals of European ancestry \citep{siva20081000}. 
	Trait values for the individuals are simulated from a linear model with SNP effect sizes generated from a spike-and-slab distribution.
	This effect size distribution coincides with the assumption of LDpred2, tilting the simulation study in its favor.
	
	We design the synthetic datesets to cover a range of genetic architectures by varying proportions of non-null SNPs and effect sizes' relationship to allele frequencies as results of negative selection. 
	We use a random subset of 100,000 individuals to generate the GWAS summary statistics for training, 10,000 for parameter tuning, and 10,000 for calculating out-of-sample prediction performance.
	The LD matrix is estimated from the 1000G reference data.
	Details of the method implementation are summarized in the supplement Section~\ref{sec:method_implement}.
	
	Figure~\ref{fig:simulation} compares the out-of-sample prediction performance of the three PRS methods under the variety of genetic architectures.
	Unsurprisingly, since the datasets are simulated under the spike-and-slab model, correctly specified LDPred2 performs the best.
	On the other hand, PRS-Bridge's performance is remarkably close despite its misspecification of the effect size distribution.
	In particular, PRS-Bridge outperforms PRS-CS consistently and by substantial margins.
	The flexibility provided by the bridge prior's exponent parameter $\alpha$ appears to play a major role in its superior performance.
	In fact, as the true causal SNP proportions decrease, so does the optimal $\alpha$ selected during the tuning process reflecting the role of $\alpha$ in controlling the prior sparsity level.

	\begin{figure}[htb]
		\centering
		\includegraphics[width=12cm,trim=0 1ex 0 .7ex,clip]{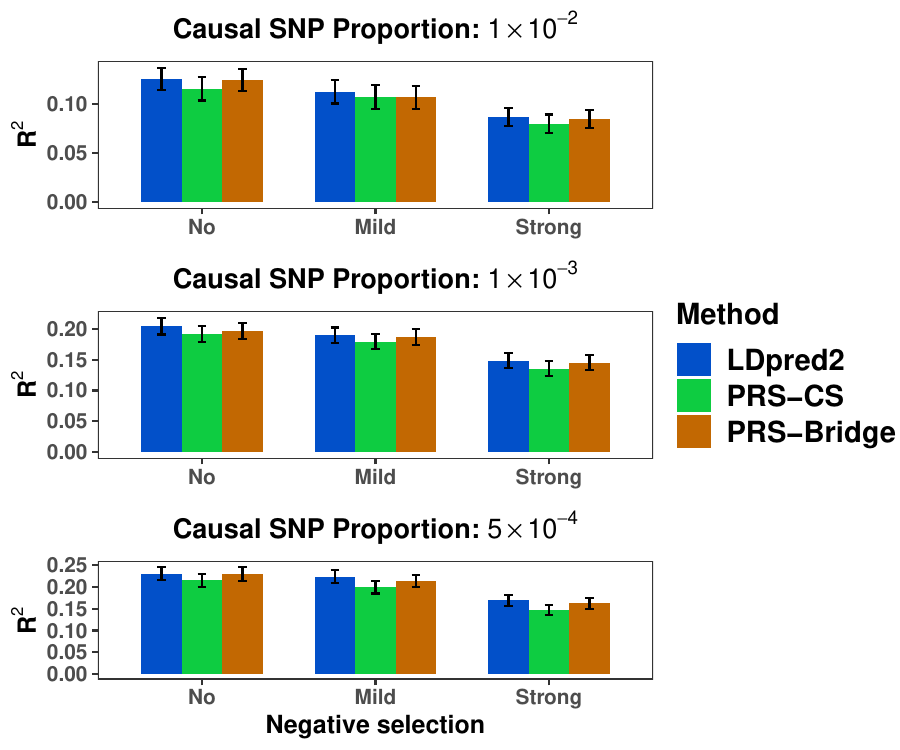}
		\caption{%
		Comparison of out-of-sample prediction performances by LDpred2, PRS-CS, and PRS-Bridge on the plasmode synthetic datasets.
		We report the average as well as 1.96 times the standard error of $R^2$ across the 100 replications. 
		The causal SNP proportions are varied from 0.01, 0.001, to 0.0005.
		The effect sizes of causal variants are assumed to be related to allele frequency under a model with no, mild, or strong negative selection. 
		The LDpred2 software uses a banded LD structure with the default LD radius of 3cM. 
		The PRS-CS and PRS-Bridge implementations use the block-diagonal LD approximation, with PRS-Bridge additionally using low-rank approximation.}
		\label{fig:simulation}
	\end{figure}
	
	\subsection{Real Data Benchmark on Continuous Traits from UK Biobank}\label{Real_data_continuous}
	We now apply the three Bayesian PRS methods and additionally Lassosum to data from UK Biobank to develop PRS for six continuous human traits: BMI, resting heart rate, high-density and low-density lipoprotein cholesterol, apolipoprotein A1, and apolipoprotein B.
	We take 319,342 unrelated European individuals from UK Biobank and use 90/5/5\% of the data as: a training set of 287,285 individuals to obtain summary data; a tuning set of 16,028 individuals to select tuning parameters; and a validation set of 16,029 individuals to calculate out-of-sample prediction $R^2$.
	We adjust for age, gender, and the top 10 genetic principal components besides genotype.
	To assess how the choice of external reference data affects PRS methods' performance, we implement each method using two different data sources, UK Biobank or 1000G, for estimating the LD matrix.

	Figure~\ref{fig:UKBiobank} summarizes the benchmark results on continuous traits.
	PRS-Bridge shows the best overall performance under the appropriately chosen block sizes: the larger blocks when using the larger reference data from UK Biobank and the smaller blocks when using the smaller reference data from 1000G.
	Under the UK Biobank reference, the large-block PRS-Bridge outperforms the default-setting PRS-CS with an average $R^2$ increase of 12.22\%, the banded LDpred2 by 2.47\%, the large-block LDpred2 by 2.72\%, and the default-setting Lassosum by an average of 14.55\% in $R^2$. 
	Under the 1000G reference, the small-block PRS-Bridge exhibits similar performance to PRS-CS and Lassosum, and outperforms the banded LDpred2 with an average $R^2$ increase of 18.5\%.
	We further compare the computational efficiency of these methods in the supplement Section~\ref{sec:speed}.

	\begin{figure}[htb]
		\centering
		\includegraphics[width=15cm]{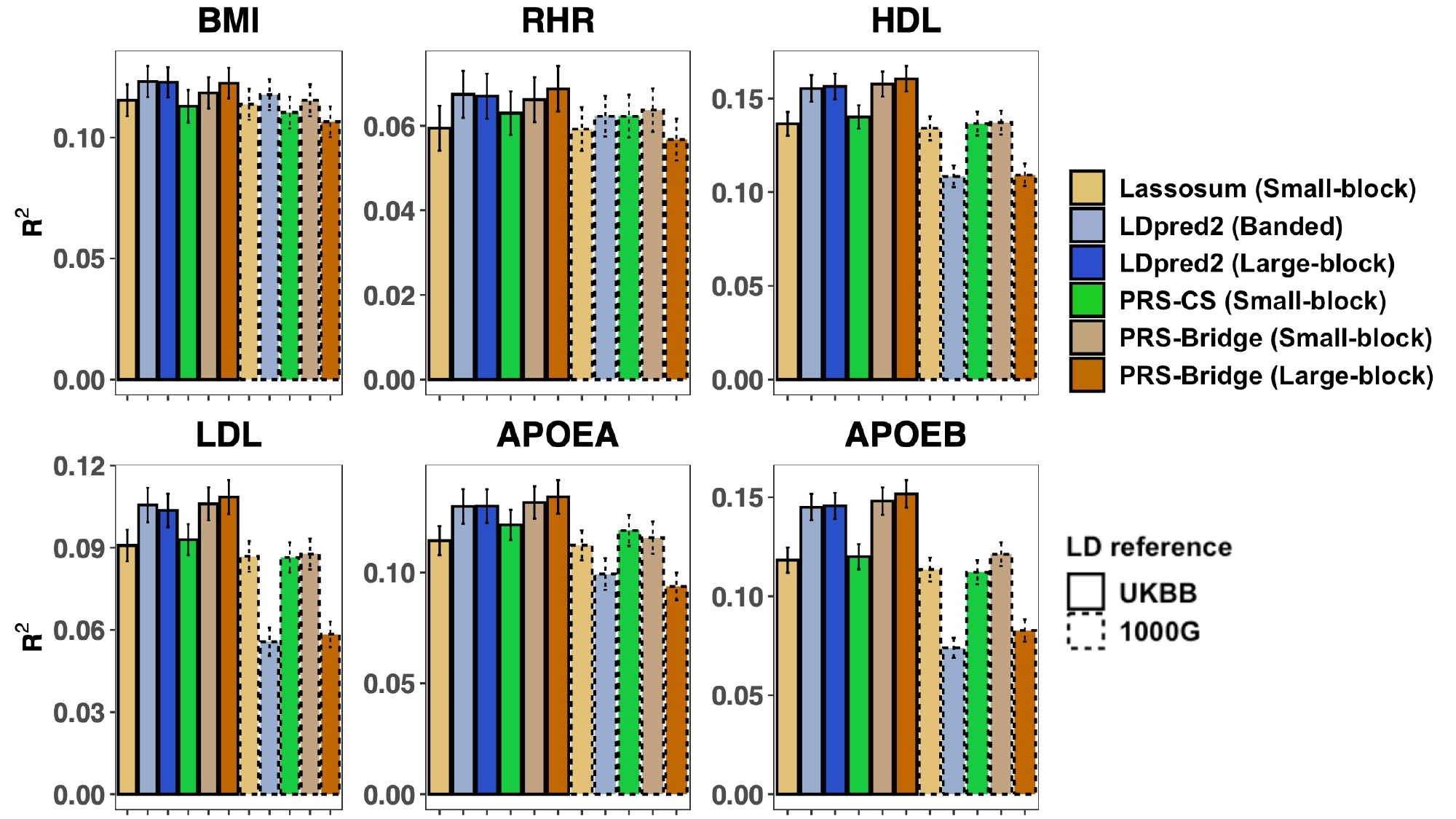}
		\caption{Out-of-sample prediction $R^2$ of Lassosum, LDpred2, PRS-CS, and PRS-Bridge on the six continuous traits: BMI, resting heart rate (RHR), high-density lipoprotein cholesterol (HDL), low-density lipoprotein cholesterol (LDL), apolipoprotein A1 (APOEA), and apolipoprotein B (APOEB). 
		We implement each method with two alternative LD reference data sources: 1000G and UK Biobank.
		``(Banded)'' indicates the use of the banded structure in approximating LD, while ``(Small-block)'' and ``(Large-block)'' indicate the use of the block structure with the small and large blocks.
		For Lassosum, LDpred2 and PRS-CS, we only consider the default LD structures in their software.
		The error bars represent 1.96 times the standard error of $R^2$ across the 100 replications. 
		Since the errors are correlated, overlaps in error bars should not be interpreted as implying the lack of statistically meaningful differences;
		Figure~\ref{fig:UKBiobank_RE} in the supplement Section~\ref{sec:relative_performance} shows $R^2$ of each method relative to Lassosum and indicates a clear trend in their relative performances that remains consistent across the replications.
		}
		\label{fig:UKBiobank}
	\end{figure}
	\FloatBarrier
	
	When using 1000G as reference, the prediction power of all methods decreases compared to when using UK Biobank.
	This is unsurprising since 1000G provides far fewer samples than UK Biobank ($N\approx500$ vs. $N \approx 330K$) for estimating LD.
	Also, the LD structure of the 1000G population may not entirely align with that of the UK Biobank population used for training and validation, contributing to the lower prediction accuracy.
	PRS-Bridge, PRS-CS, and Lassosum nonetheless show performance more robust to the choice of LD reference samples than LDpred2.
	Our results also show that the smaller LD reference sample size warrants more regularization in estimation of the LD matrix by using a smaller block size.
	When using UK Biobank as reference, the large-block PRS-Bridge on average yields 2.54\% improvement in prediction over the smaller-block one.
	On the other hand, when using 1000G as reference, the small-block PRS-Bridge outperforms the larger-block one, with 20.54\% improvement on average.

	\subsection{Real Data Benchmark on Binary Traits using Summary Statistics from External Sources}
	\label{Real_data_binary}
	Here we benchmark the PRS methods on five disease traits: breast cancer, coronary artery disease, depression, rheumatoid arthritis, and inflammatory bowel disease. 
	For binary disease outcomes, even the largest biobanks such as UK Biobank lack an adequate number of cases and sufficient statistical power. 
	For this reason, we train the models on publicly available GWAS summary data from external sources and evaluate their predictive performances on individual-level validation samples from UK Biobank. 
	Table~\ref{tab:GWAS_disease} lists the sources of summary data, along with sample sizes of the training, tuning, and validation datasets.
	
	
	Figure~\ref{fig:Disease} summarizes the results of our benchmark study on binary disease traits.
	The results here show the same patterns as in our study on continuous traits:
	PRS-Bridge generally outperforms Lassosum, LDpred2 and PRS-CS;
	the larger LD reference sample size leads to better performances;
	and PRS-Bridge, PRS-CS, and Lassosum are more robust than LDpred2 to the choice of reference data.
	PRS-Bridge's predictive ability for inflammatory bowel disease is particularly remarkable.
	When using UK Biobank as reference, the larger-block PRS-Bridge delivers 25.2\% improvement in predicting inflammatory bowel disease over the best performing LDpred2, 27.27\% improvement over PRS-CS, and 10.2\% improvement over Lassosum.

\begin{figure}[ht]
	\centering
	\includegraphics[width=15cm]{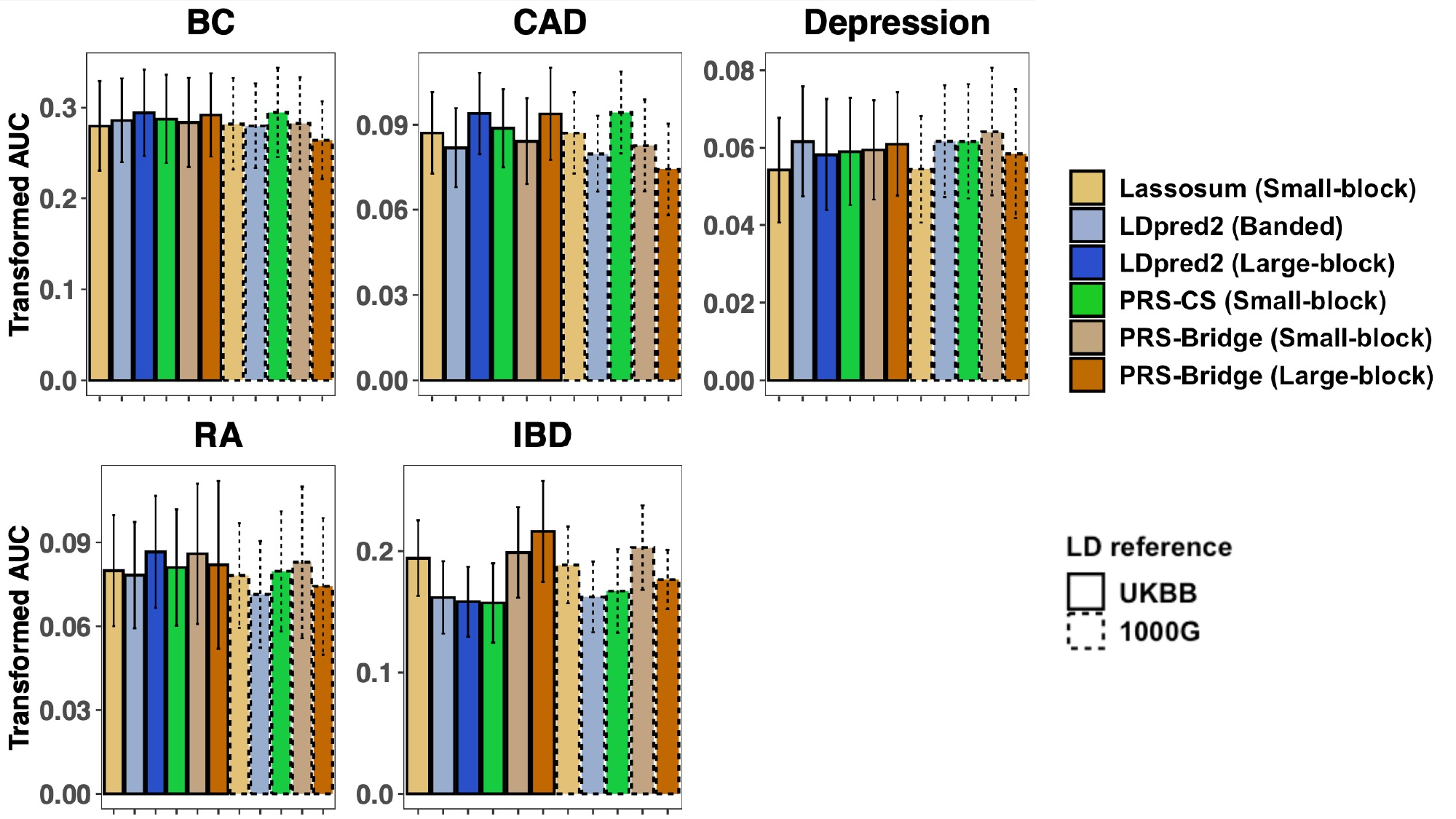}
	\caption{Out-of-sample prediction performances of Lassosum, LDpred2, PRS-CS, and PRS-Bridge on the five binary disease traits: breast cancer (BC), coronary artery disease (CAD), depression, rheumatoid arthritis (RA), and inflammatory bowel disease (IBD).
	The methods are implemented in the same manners as described in the caption of Figure~\ref{fig:UKBiobank}.
	The error bars should again be interpreted with caution; Figure~\ref{fig:disease_RE} in the supplement Section~\ref{sec:relative_performance} provides the transformed AUC of each method relative to Lassosum.}
	\label{fig:Disease}
\end{figure}
\FloatBarrier

\subsection{Additional Benchmark Results}
\label{sec:additional_num_results}

	We conduct additional benchmark studies to shed lights on precise sources of the differences in performance across the methods. 
	We defer the results to the supplement Section~\ref{sec:add_num_result} but summarize each study and its findings here.
	
	We investigate, in the supplement Section~\ref{sec:PRS-Bridge-tuning}, which exponent parameter maximizes PRS-Bridge's predictive performance for each trait;
	the optimal value provides an insight on the trait's genetic architecture.
	We also assess how PRS-Bridge's performances change as we vary the proportion of eigenvectors projected away from the LD matrix.
	Additionally, we investigate whether PRS-CS can provide a similar adaptability to different genetic architectures, using the two other choices of priors explored by \cite{ge2019polygenic}.
	We find that this additional tuning makes relatively little difference in PRS-CS's performance, with the default prior having the most consistent performance.
	
	In the supplement Section~\ref{sec:PRS-CS-proj-binary}, we assess how PRS-CS performs when replacing its ad hoc constraint on the prior variance with our summary statistics projection. 
	This study demonstrates our projection approach to be applicable beyond PRS-Bridge.
	Also, comparing PRS-CS with projection against PRS-Bridge and the original PRS-CS allows us to separate the effect of the prior choice from that of the projection.
	We find PRS-CS with projection performs comparably to the original PRS-CS in some cases but slightly worse in other cases.
	This shows that, besides the proper inference achieved by the projection, the prior choice does play a major role in PRS-Bridge's superior performance.
	
	We also benchmark PRS-CS with projection against PRS-CS with modifications of the ad hoc constraint $\tau^2 \lambda_j^2 \leq \priorSdBd^2$, where we vary $\priorSdBd = 10^k/\sqrt{\Ntrain}$ for $k=-2,-1,0 \text{ (original)},\allowbreak 1, 2$ (Supplement Section~\ref{sec:PRS-CS-proj}). 
	The results confirm that the performance of PRS-CS with ad hoc constraint is sensitive to the choice of $\priorSdBd$, while our projection approach yields more consistent and competitive performances.
	We additionally explore the potential use of $(\bm{I} - \bm{P}_\textrm{ref}) \, \bsumStats$, the part of the summary statistics orthogonal to reference LD, as discussed in Section~\ref{lowRankApprox}.
	Specifically, we find that performance of PRS-CS with projection is improved when we incorporate the information from the orthogonal component; 
	this can be done in a manner that emulates the effect of the ad hoc constraint based on $\priorSdBd = 1 / \Ntrain$, which we confirm to yield comparable performances to the original PRS-CS.
	
	Finally, in the supplement Section~\ref{sec:PRS-Bridge-auto-numerical}, we evaluate performances of the tuning-free, auto version of PRS-Bridge as mentioned in Section~\ref{PRSBridge}.
	Unsurprisingly, the auto version lags behind the tuned version but only by small margins for continuous traits, where data provide more information.

\section{Discussion}\label{discussion}
This work advances the Bayesian PRS methodology on multiple fronts.
We have uncovered a previously overlooked pitfall in the widely used approach based on the approximate likelihood and proposed a principled remedy via the projection of GWAS summary statistics. 
While we have focused on applications involving PRS development, the issue is also relevant for Bayesian methods for other applications like fine-mapping of causal variants within genomic regions \citep{wang2020simple,yang2023carma}.
We have further demonstrated, through an extensive numerical study, an opportunity to improve PRS performance by more carefully considering the choice of a prior on SNP effect size distribution, of reference data for estimating LD, and of a structured approximation of LD.

Our proposed PRS-Bridge delivers consistent and superior performance thanks to its use of the projected summary statistics and the bridge prior's flexibility to adapt to varying effect size distributions. 
We choose the bridge prior for its computational advantage but other flexible priors, such as the normal-gamma priors \citep{brown2010inference}, can in principle provide similar benefits in terms of predictive performance.
In our real data analyses, PRS-Bridge demonstrates a meaningful increase in predictive power over LDPred2, PRS-CS, and Lassosum for a number of clinically relevant traits such as inflammatory bowel disease. 
Incidentally, we note that our PRS-Bridge is distinct from the similarly named BridgePRS of \cite{hoggart2024bridgeprs}.
BridgePRS is so called because it ``bridges'' multiple populations for cross-ancestry PRS construction, but the method actually relies on a Gaussian prior and \textit{not} on the bridge prior.

We have also shown that 
the projection approach can be extended to incorporate the part of the summary statistics orthogonal to the column space of reference LD (Section~\ref{lowRankApprox} and Supplement Section~\ref{sec:model_null_space}). 
This potential extension is relevant because our empirical findings of Section~\ref{Numerical_studies} and of the supplement Section~\ref{sec:add_num_result} indicate that the part discarded by the projection indeed contains relevant information on the target population, with the amount of the discarded information being more substantial when the reference data source is limited in its sample size and when its underlying population differs from the target.
The connection between the auxiliary modeling of the orthogonal component and the LD matrix regularization further highlights the importance of how to select the reference data source and how to resolve its discrepancy with the training data.

As PRS methodology moves in the direction of joint modeling across diverse ancestry groups and related traits, we expect the use of a flexible prior to become even more important in capturing a range of multivariate SNP effect size distributions. 
Another direction is the choice of prior that will allow flexible incorporation of functional annotation information on genetic variants derived from external data.
Bayesian approaches to PRS have much potential, and further advances in the methodology can help realize the full potential of genomic medicine.

Looking ahead, other types of high-dimensional biomarkers—such as proteomic profiles—are expected to play an increasingly important role in enhancing risk prediction for complex diseases. 
The use of summary-statistics data is similarly applicable to these settings and can provide larger sample sizes for model development. 
However, the architectures of proteomic effects may differ substantially from the polygenic patterns observed in GWAS, which could make the ad hoc techniques developed and fined tuned for GWAS less successful in these applications. 
A flexible and principled approach like PRSBridge, on the other hand, is likely to generalize more straightforwardly to other future applications of summary-level data.

\section*{Data and Code Availability}
PRS-Bridge is freely available as a Python-based command line tool from a GitHub repository at \url{https://github.com/YuzhengDun1999/PRSBridge}.
Code and instructions to reproduce the numerical results is available at \url{https://github.com/YuzhengDun1999/PRS-Bridge-article-code}.







\section*{Supplementary material}
The Supplementary Material presents the followings:
proofs of the theoretical results and their real data demonstrations; 
the extension of the projection approach as mentioned in Section~\ref{lowRankApprox};
the auto version of PRS-Bridge;
details of the PRS methods' implementations and of the data quality control;
a comparison of the methods' computational speeds; 
and the additional benchmark studies as summarized in Section~\ref{sec:additional_num_results}.

\section*{Acknowledgments}
{\sloppy
Research reported in this publication was supported by 
the National Human Genome Research Institute, 
the National Cancer Institute,
and the National Institute of General Medical Sciences
of the National Institutes of Health under Award Numbers 
R01HG010480 (N.C.), U01CA249866 (N.C.),
R00HG012223 (J.J.), R35GM157133 (J.J.), and
R35GM160458 (A.N.).
The content is solely the responsibility of the authors and does not necessarily represent the official views of the National Institutes of Health.
\par}

Individual-level genotype and phenotype data of the UK Biobank samples were obtained under application 17731.

\bibliographystyle{agsm}
\bibliography{arxiv}

\newpage
\renewcommand{\thetheorem}{S\arabic{theorem}} 
\renewcommand{\thelemma}{S\arabic{lemma}} 
\theoremstyle{definition}
\renewcommand{\thedefinition}{S\arabic{definition}} 
\renewcommand{\theassumption}{S\arabic{assumption}} 
\renewcommand{\theequation}{S.\arabic{equation}}
\renewcommand{\thesection}{S\arabic{section}}
\renewcommand{\thetable}{S\arabic{table}}
\renewcommand{\thefigure}{S\arabic{figure}}
\renewcommand{\theequation}{S\arabic{equation}}
\renewcommand{\thepage}{S\arabic{page}}
\setcounter{section}{0}
\setcounter{page}{1}
\setcounter{equation}{0}
\setcounter{figure}{0}
\setcounter{theorem}{0}
\setcounter{table}{0}

\bigskip
\begin{center}
	{\huge\bf Supplementary Material}
\end{center}
\section{Proof of Theorem~\ref{thm2}}
\label{proof2}

To state in a mathematically precise manner Theorem~\ref{thm2} in the main text and Theorem~\ref{thm1} in Section~\ref{proof1} below, we need to make a technical assumption.
Assumption~\ref{continuous} below ensures that a summary statistics $\bsumStats = \Ntrain^{-1} \bm{X}_\mathrm{train}^T \bm{y}_\mathrm{train} $ almost surely has non-zero component in every direction of $\colsp(\ld_{\mathrm{train}}) = \colsp(\bm{X}_{\mathrm{train}}^T) $;
i.e.\ for any fixed non-zero vector $\bm{v} \in \colsp(\ld_{\mathrm{train}})$,  we have $\langle \bsumStats, \bm{v} \rangle \neq 0$ with probability one.
\begin{assumption}\label{continuous}
	The probability distribution $\bm{y}_\mathrm{train} \given \bm{X}_\mathrm{train}$ underlying the individual-level data is absolutely continuous with respect to the Lebesgue measure in $\mathbb{R}^N$.
\end{assumption}

A heavy-tailed distribution can intuitively thought of as a distribution whose density has tails that decay slower than exponentials.
More formally, it is defined as follows:
\begin{definition}[Heavy-tailed distribution]\label{heavy_tailed_def}
	A distribution with density $\pi(\cdot)$ and cumulative distribution $F(\cdot)$ is said to be \textit{heavy-tailed} if $\int_{-\infty}^{\infty}\pi(x)\exp(tx)dx=\infty$ for any $t > 0$ or, equivalently, if $\lim_{x\rightarrow\infty}\exp(tx)(1- F(x)) = \infty$.
\end{definition}

We first establish three lemmas to be use for our proof of Theorem~\ref{thm2}:
\begin{lemma}\label{lemma1}
	Suppose $\pi(\beta)$ admits a scale mixture representation $\pi\left(\beta\given\priorSD^2\right)\sim\mathcal{N}(0, \priorSD^2)$ with $\priorSD^2 \sim\pi(\priorSD^2)$. 
	Then $\pi(\beta)$ is heavy-tailed if and only if $\pi(\priorSD^2)$ is heavy-tailed.
\end{lemma}
\begin{proof}
	By Tonelli's theorem, we have for any $t > 0$
	\begin{equation}\label{heavy-tailed}
		\begin{split}
			&\hspace*{-2em}\int_{-\infty}^{\infty}\pi\left(\beta\right)\exp\left(t\beta\right)d\beta\\
			&=\int_{\infty}^{\infty}\int_{0}^{\infty}\pi\left(\beta\given\priorSD^2\right)\pi\left(\priorSD^2\right)\exp\left(t\beta\right)d\priorSD^2 d\beta\\
			&=\int_{0}^{\infty}\int_{\infty}^{\infty}\pi\left(\beta\given\priorSD^2\right)\pi\left(\priorSD^2\right)\exp\left(t\beta\right)d\beta d\priorSD^2 \\
			&=\int_{0}^{\infty}\int_{\infty}^{\infty}\left(2\pi\priorSD^2\right)^{-1/2}\pi\left(\priorSD^2\right)\exp\left(-\frac{\beta^2-2\priorSD^2 t\beta}{2\priorSD^2}\right)d\beta d\priorSD^2 \\
			&=\int_{0}^{\infty}\pi\left(\priorSD^2\right)\exp\left(\frac{\priorSD^2 t^2}{2}\right)d\priorSD^2
		\end{split}
	\end{equation}
	Hence $\int_{-\infty}^{\infty}\pi\left(\beta\right)\exp\left(t\beta\right)d\beta=\infty$ if and only if $\int_{0}^{\infty}\pi\left(\priorSD^2\right)\exp\left(\frac{\priorSD^2 t^2}{2}\right)d\priorSD^2=\infty$.
\end{proof}

\begin{lemma}\label{lemma2}
	If a density $\pi\left(x\right)$ on $x \geq 0$ is heavy-tailed, then $\pi'\left(x\right)\propto \left(1 + \eigenval x \right)^{-1/2}\pi\left(x\right)$ is also heavy-tailed for any $\eigenval>0$.
\end{lemma}
\begin{proof}
	Denote the reciprocal of the normalizing constant by $C = \left( \int \pi'(x) dx \right)^{-1}$ and let $t > 0$.
	Observe that 
	\begin{equation}\label{heavy-tailed-transfered}
		\begin{split}
			&\hspace*{-2em}\int_{0}^{\infty}\pi'\left(x\right)\exp\left(tx\right)dx\\
			&= C\int_{0}^{\infty}\left(\eigenval x+1\right)^{-1/2}\pi\left(x\right)\exp\left(tx\right)dx\\
			&= C\int_{0}^{\infty}\pi\left(x\right)\exp\left(\frac{tx}{2}\right)\frac{\exp\left(\frac{tx}{2}\right)}{\sqrt{\eigenval x+1}}dx\\
			& \geq C\int_{L}^{\infty}\pi\left(x\right)\exp\left(\frac{tx}{2}\right)dx.\\
		\end{split}
	\end{equation}
	The last inequality comes from the fact that there exists $L < \infty$ such that $\exp\left(tx/2\right)>\left(\eigenval x+1\right)^{1/2}$ for all $x > L$. 
	Since $\pi\left(x\right)$ is heavy-tailed and
	$\int_{0}^{\infty}\pi\left(x\right)\exp\left(\frac{tx}{2}\right)dx=\infty$, we must also have $\int_{L}^{\infty}\pi\left(x\right)\exp\left(\frac{tx}{2}\right)dx = \infty$.
\end{proof}

\begin{lemma}\label{lemma3}
	Suppose that  $X_1,X_2,...,X_n$ are i.i.d.\ random variables with cumulative distribution $F_X(\cdot)$. 
	Let $Y=\min\left(X_1,X_2,...,X_n\right)$ and denote its cumulative distribution by $F_Y(\cdot)$. 
	If $F_X(\cdot)$ is heavy-tailed, then so is $F_Y(\cdot)$. 
\end{lemma}
\begin{proof}
	Since the distribution of $X_1,X_2,...,X_n$ is heavy-tailed, by definition we have 
	$$\lim_{x\rightarrow\infty}\exp\left(tx\right)\left(1-F_X\left(x\right)\right)=\infty \ \text{ for any } \, t>0.$$
	It follows that, for any $n>0$,
	$$\lim_{x\rightarrow\infty}\exp\left(ntx\right)\left(1-F\left(x\right)\right)^n=\left(\lim_{x\rightarrow\infty}\exp\left(tx\right)\left(1-F\left(x\right)\right)\right)^n=\infty \ \text{ for any } \, t>0.$$
	Since the cumulative density of $Y=\min\left(X_1,...,X_n\right)$ is given by
	$F_Y\left(y\right)=1-\left[1-F_X\left(y\right)\right]^n$, we have for any $s>0$
	$$\lim_{y\rightarrow\infty}\exp\left(sy\right)\left(1-F_Y\left(y\right)\right)
	=\lim_{y\rightarrow\infty}\exp\left(n t y\right)\left[1-F_X\left(y\right)\right]^n=\infty 
	\ \text{ for } \, t = s / n > 0.$$
	As this holds for any $s>0$, the distribution of $Y$ is heavy-tailed.
\end{proof}

We proceed to the proof of Theorem \ref{thm2}:

\begin{proof}[Proof of Theorem \ref{thm2}]
	Consider a prior $\beta_j\given\tau,\priorSD_j\sim\mathcal{N}\left(0,\tau^2\priorSD_j^2\right)$ with $\priorSD^2_j\sim \pi\left(\priorSD^2_j\right)$. Denote $\bLambda=\text{diag}\left(\lambda_1,...,\lambda_P\right)$ and $\bm{\Phi}=\Ntrain\ldRef+\tau^{-2}\bLambda^{-2}$. 
	Setting $\hat{\bbeta}=\bm{0}$ in Equation (\ref{joint_post}), we obtain the following relation for the joint nominal posterior:
	\begin{equation}
		\begin{split}
			&\hspace*{-2.5em}\pi\left(\bbeta,\priorSDs\given \sumStatsTrain,\ldRef,\tau\right)\\
			&\propto \frac{\exp\left(-\frac{\left(\bbeta-\Ntrain\bm{\Phi}^{-1}\sumStatsTrain\right)^T\bm{\Phi}\left(\bbeta-\Ntrain\bm{\Phi}^{-1}\sumStatsTrain\right)}{2}\right)\prod_{j=1}^P\priorSD_j^{-1}\pi\left(\priorSD_j^2\right)}{\exp\left(-\frac{\left(\Ntrain\bm{\Phi}^{-1}\sumStatsTrain\right)^T\bm{\Phi}\left(\Ntrain\bm{\Phi}^{-1}\sumStatsTrain\right)}{2}\right)}\\
			&= \exp\left(-\frac{\bbeta^T\bm{\Phi}\bbeta-2\bbeta^T\Ntrain\sumStatsTrain}{2}\right)\prod_{j=1}^P\priorSD_j^{-1}\pi\left(\priorSD_j^2\right)\\
			&= (2\pi)^{-P/2}\varphi\left(\bbeta\given\bm{\Phi}^{-1}\Ntrain\sumStatsTrain,\bm{\Phi}^{-1}\right) \\
			&\hspace*{3em}\operatorname{det}\left(\bm{\Phi}\right)^{-1/2}\exp\left(\frac{\Ntrain^2\left(\sumStatsTrain\right)^T\bm{\Phi}^{-1}\sumStatsTrain}{2}\right)\prod_{j=1}^P\priorSD_j^{-1}\pi(\priorSD_j^2),
		\end{split}
	\end{equation}
	where $\varphi\left(\bbeta\given\bm{\mu},\bm{\Sigma}\right)$ denotes the normal probability density function with mean $\bm{\mu}$ and variance $\bm{\Sigma}$ evaluated at $\bbeta$.
	In addition, we have
	\begin{equation}\label{Phi_det_inequal}
		\begin{split}
			&\hspace*{-2em}
			\tau^{-1}\text{det}\left(\bm{\Phi}\right)^{-1/2}\prod_{j=1}^P\priorSD_j^{-1} \\
			&=\text{det}\left(\Ntrain\tau^2\bLambda\ldRef\bLambda+\bm{I}\right)^{-1/2}\\
			&= \text{det}\left(\Ntrain\tau^2\bLambda\VRef\eigenvalRef\VRef^T\bLambda+\bm{I}\right)^{-1/2}\\
			&\geq\text{det}\left(\Ntrain\eigenval_{\text{ref,max}}\tau^2\bLambda^2+\bm{I}\right)^{-1/2}\\
			&= \prod_{j=1}^P\left(\Ntrain\tau^2\priorSD_j^2\eigenval_{\text{ref,max}}+1\right)^{-1/2},
		\end{split}
	\end{equation}
	where $\eigenval_{\text{ref,max}}$ is the largest eigenvalue of $\ldRef$;
	the inequality of the determinants follow from the fact 
	$\left(\Ntrain\tau^2\bLambda\VRef\eigenvalRef\VRef^T\bLambda+\bm{I}\right) \succ \left(\Ntrain\eigenval_{\text{ref,max}}\tau^2\bLambda^2+\bm{I}\right)$
	and Weyl's inequality \citep{horn2012matrix-analysis}.
	
	Denote by $\pi'(\priorSD^2_j)$ the density proportional to $\left(\Ntrain\tau^2\priorSD_j^2\eigenval_{\text{ref,max}}+1\right)^{-1/2}\pi\left(\priorSD^2_j\right)$. 
	Also denote $\priorSDmin^2=\min\left(\priorSD_1^2,...,\priorSD^2_P\right)$.  
	Now observe that
	\begin{equation}\label{joint_inequality}
		\begin{split}
			&\hspace*{-2em} \int\pi\left(\bbeta,\priorSDs\given \sumStatsTrain,\ldRef,\tau\right)d\bbeta d\priorSD^2_1...d\priorSD^2_P\\
			&= C_1\int\text{det}\left(\bm{\Phi}\right)^{-1/2}\exp\left(\frac{\Ntrain^2\left(\sumStatsTrain\right)^T\bm{\Phi}^{-1}\sumStatsTrain}{2}\right)\prod_{j=1}^P\priorSD_j^{-1}\pi\left(\priorSD^2_j\right)d\priorSD^2_1...d\priorSD^2_P\\
			&\geq C_1\int\exp\left(\frac{\Ntrain^2(\sumStatsTrain)^T\bm{\Phi}^{-1}\sumStatsTrain}{2}\right)\prod_{j=1}^P\left(\Ntrain\tau^2\priorSD_j^2\eigenval_{\text{ref,max}}+1\right)^{-1/2}\pi\left(\priorSD^2_j\right)d\priorSD^2_1...d\priorSD^2_P\\
			&= C_2\int\exp\left(\frac{\Ntrain^2(\sumStatsTrain)^T\bm{\Phi}^{-1}\sumStatsTrain}{2}\right)\prod_{j=1}^P\pi'\left(\priorSD^2_j\right)d\priorSD^2_1...d\priorSD^2_P\\
			&\geq C_3\int\exp\left(\frac{\Ntrain(\sumStatsTrain)^T\VRef(\eigenvalRef+\Ntrain^{-1}\tau^{-2}\priorSDmin^{-2}\bm{I})^{-1}\VRef^T\sumStatsTrain}{2}\right)\\
			&\hspace*{1cm}\pi'\left(\priorSDmin^2\right)\pi'\left(\priorSD^2_1,...,\priorSD^2_P\given \priorSDmin^2\right)d\priorSD^2_1...d\priorSD^2_P\\
			&\geq C_4\int\exp\left(\frac{\Ntrain\langle\bv_{\text{ref},k},\sumStatsTrain\rangle^2}{2(\eigenval_{\text{ref},k}+\Ntrain^{-1}\tau^{-2}\priorSDmin^{-2})}\right)\pi'(\priorSDmin^2)d\priorSDmin^2,
		\end{split}
	\end{equation}
	where the first inequality follows from Equation \eqref{Phi_det_inequal}, the second inequality follows from the fact $\bm{\Phi} \prec \VRef\left(\Ntrain\eigenvalRef+\tau^{-2}\priorSDmin^{-2}\bm{I}\right)\VRef^T$ and holds for any $k$, and $C_1,C_2,C_3, C_4$ are positive constants. 
	The same logic as in the proof of Theorem~\ref{thm1} tells us when $\operatorname{null}( \ld_{\text{ref}}) \not\subseteq \operatorname{null}(\ld_{\text{train}})$, there exists an eigenvalue/vector pair $\left(\eigenval_{\text{ref},k},\bv_{\text{ref},k}\right)$, such that $\eigenval_{\text{ref},k}=0$ and $\langle\bv_{\text{ref},k},\sumStatsTrain\rangle\neq0$ for almost every realization of summary statistics.
	Now, by Lemma \ref{lemma1}, \ref{lemma2}, and \ref{lemma3}, we know that $\pi'(\priorSDmin^2)$ is heavy-tailed, and hence
	\begin{equation}\label{heavy_euqality}
		\int\exp\left(\frac{\Ntrain^2\tau^2\langle\bv_{\text{ref},k},\sumStatsTrain\rangle^2\priorSDmin^2}{2}\right)\pi'\left(\priorSDmin^2\right)d\priorSDmin^2 = \infty.
	\end{equation}
	Combining Equation (\ref{joint_inequality}) and (\ref{heavy_euqality}), we conclude that the joint nominal posterior is improper.
\end{proof}

\section{Additional Theoretical Result on and Real Data Demonstration of Nominal Posterior's Ill-behavior}\label{proof1}
In this section, we introduce an additional theoretical result that complements Theorem~\ref{thm2} in the main text by demonstrating another problematic behavior in the nominal posterior.

\begin{theorem}\label{thm1}
	The following result holds for almost every realization of $\bsumStats$ under Assumption~\ref{continuous}. 
	Consider the posterior distribution $\bm{\beta} \given \bsumStats, \bm{D}_{\mathrm{ref}},  \sigma_0$ based on the prior $\bm{\beta} \given \sigma_0 \sim \mathcal{N}\left(\bm{0}, \sigma_0^2 \bm{I}\right)$ with standard deviation $\sigma_0 > 0$.
	If $\operatorname{null}( \bm{D}_{\mathrm{ref}} ) \not\subseteq \operatorname{null}(\bm{D}_{\mathrm{train}})$, 
	then the posterior mean $\hat{\bm{\beta}}\left(\sigma_0\right) = \mathbb{E}\left[ \bm{\beta} \given  \bsumStats, \bm{D}_{\mathrm{ref}}, \sigma_0 \right]$ 
	diverges in the limit of the prior becoming uninformative; 
	i.e.,
	$\left\| \hat{\bm{\beta}}(\sigma_0) \right\|^2 \to \infty$ as $\sigma_0 \to \infty$. 
	No such pathological behavior occurs when $ \bm{D}_{\mathrm{ref}} = \bm{D}_{\mathrm{train}} $ or, more generally, when $\operatorname{null}\left( \bm{D}_{\mathrm{ref}} \right) \subseteq \operatorname{null}\left(\bm{D}_{\mathrm{train}}\right)$.
\end{theorem}

Ordinarily, when fitting a linear regression model of the form \eqref{lm_model} under the prior $\bm{\beta} \given \sigma_0 \sim \mathcal{N}\left(\bm{0}, \sigma_0^2 \bm{I}\right)$, the posterior mean converges to a finite limit as $\sigma_0 \to \infty$;
this is true even when $N < p$. 
Theorem~\ref{thm1} shows that the nominal posterior \eqref{nominal_post} defies this expected behavior and pulls the posterior estimate towards infinity when there is a mismatch between the LD reference and GWAS training samples.
The analysis in our proof below reveals that this is caused by the approximate likelihood \eqref{likelihood_sum} being falsely informative in the null space of $\bm{D}_{\text{ref}} = \bm{X}_\textrm{ref}^T \bm{X}_\textrm{ref}/N_\textrm{ref}$.

\begin{proof}
	Let $\ldRef = \VRef \eigenvalRef \VRef^T$ denote the eigenvalue decomposition of the LD correlation matrix estimated from external reference data.
	To establish the claimed behaviors of the nominal posterior, we first express the posterior mean in terms of the eigenvalues and eigenvectors of $\ldRef$ as
	\begin{align*}
		\hat{\bm{\mu}}
		&= \left(\Ntrain \ldRef+\sigma_0^{-2}\bm{I}\right)^{-1}\Ntrain\sumStatsTrain \\
		&= \VRef\left(\eigenvalRef+\frac{1}{\Ntrain}\sigma_0^{-2}\bm{I}\right)^{-1}\VRef^T\sumStatsTrain \\
		&= \sum_{\ell=1}^{P}\frac{\langle\bv_{\text{ref},\ell},\sumStatsTrain\rangle}{\eigenval_{\text{ref},\ell}+\frac{1}{\Ntrain}\sigma_0^{-2}}\bv_{\text{ref},\ell},
	\end{align*}
	where $\eigenval_{\text{ref},\ell}\geq0$ and $\bv_{\text{ref},\ell}$ denote the eigenvalues and eigenvectors of $\ld_{\text{ref}}$.
	We can hence express the squared norm of $\hat{\bm{\mu}}$ as
	\begin{equation}\label{mu_t_mu}
		\begin{split}
			\| \hat{\bm{\mu}} \|^2
			&= \sum_{l=1}^{P}\frac{\langle\bv_{\text{ref},\ell},\sumStatsTrain\rangle^2}{\left(\eigenval_{\text{ref},\ell}+\frac{1}{\Ntrain}\sigma_0^{-2}\right)^2}.
		\end{split}
	\end{equation}
	
	We now use the Equation \eqref{mu_t_mu} to study the behavior of $\| \hat{\bm{\mu}} \|$.
	We first consider the mismatched case $\operatorname{null}( \ld_{\text{ref}}) \not\subseteq \operatorname{null}(\ld_{\text{train}})$. This means that there exists an eigenvector $\bv_{\text{ref},k}$ such that $\bv_{\text{ref},k}\in\operatorname{null}(\ldRef)$, with corresponding eigenvalue $\eigenval_{\text{ref},k} = 0$, and $\bv_{\text{ref},k}\not\in\operatorname{null}( \ld_{\text{train}})$.
	The latter fact means that $\bv_{\text{ref},k}$ must have a non-zero component in $\colsp( \ld_{\text{train}} )$ since $\operatorname{null}( \ld_{\text{train}}) = \colsp( \ld_{\text{train}} )^\perp$.
	This means that, by Assumption~\ref{continuous}, we have $\langle\bv_{\text{ref},k},\sumStatsTrain\rangle\neq0$ for almost every realization of summary statistics.
	For such summary statistics, the Equation \eqref{mu_t_mu} tells us that
	$$\big\| \hat{\bm{\mu}} \big\| 
	\geq 
	\Ntrain\sigma_0^2\lvert\langle\bv_{\text{ref},k},\sumStatsTrain\rangle\rvert,$$
	where the right hand side tends to $\infty$ as $\sigma_0\rightarrow\infty$.
	
	We now consider the matched case $\operatorname{null}( \ld_{\text{ref}}) \subseteq \operatorname{null}(\ld_{\text{train}})$. 
	In this case, since $\sumStatsTrain\in\operatorname{col}(\ld_{\text{train}})=\operatorname{null}(\ld_{\text{train}})^\perp$, we have $\langle\bv_{\text{ref},k},\sumStatsTrain\rangle=0$ for every eigenvector $\bv_{\text{ref},k}\in\operatorname{null}( \ld_{\text{ref}})\subseteq \operatorname{null}(\ld_{\text{train}})$ associated with zero eigenvalue.
	Hence the Equation \eqref{mu_t_mu} becomes
	$$\big\| \hat{\bm{\mu}} \big\|^2 
	= \sum_{\ell \, : \, \eigenval_{\text{ref}, \ell} \neq 0} \frac{\langle\bv_{\text{ref},\ell},\sumStatsTrain\rangle^2}{\left(\eigenval_{\text{ref},\ell}+\frac{1}{\Ntrain}\sigma_0^{-2}\right)^2}
	\leq \sum_{\ell \, : \, \eigenval_{\text{ref}, \ell} \neq 0} \frac{\langle\bv_{\text{ref},\ell},\sumStatsTrain\rangle^2}{\eigenval_{\text{ref},\ell}^2} <\infty.$$
	The norm $\| \hat{\bm{\mu}} \|$ thus remains bounded even as $\sigma_0\rightarrow\infty$.
\end{proof}

We now demonstrate the consequence of Theorem~\ref{thm1} in the matched and mismatched cases using real data.
We use data from UK Biobank and consider BMI as an outcome and 200 consecutive SNPs from the first LD block of chromosome 22 as predictors. 
For the matched case, we obtain both summary-statistics and LD matrix from the same approximately 37K unrelated white individuals. 
For the mismatched case, we obtain the input data of the same size but from two distinct sets of individuals.
Figure \ref{post_gaussian_mismatch} plots the posterior mean as a function of $\sigma_0$ under these two scenarios, clearly demonstrating the two distinct behaviors of the posterior as predicted by the theorem.

\begin{figure}[!htb]
	\centering
	\includegraphics[width=9.5cm]{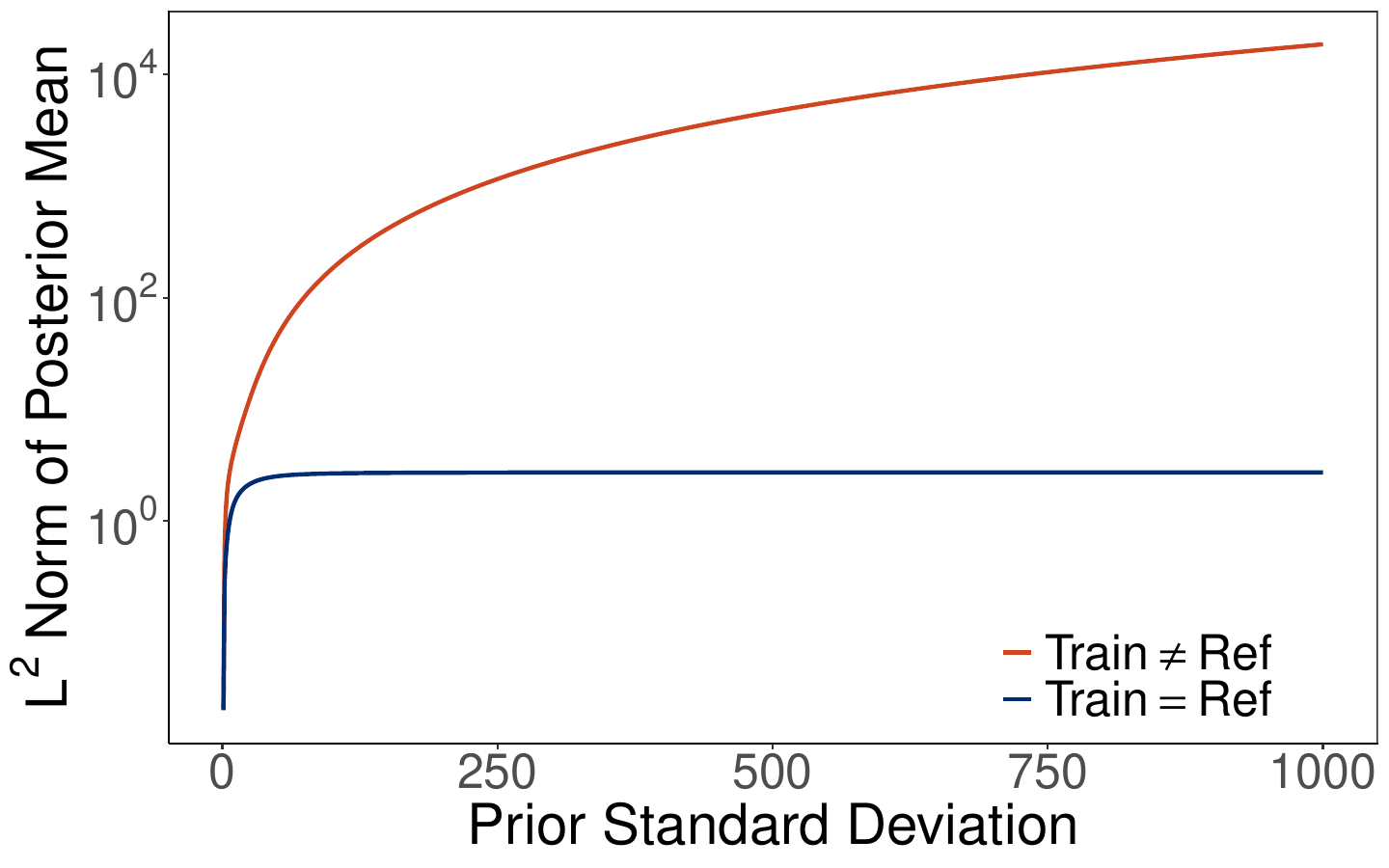}
	\caption{%
		The $L^2$ norm of the posterior mean of $\bm{\beta}$ under different values of the prior standard deviation $\sigma_0$ when the reference sample matches the training sample ($\ld_{\text{ref}} =\ld_{\text{train}}$, blue curve) and when the reference sample does not match the training sample ($ \ld_{\text{ref}} \neq \ld_{\text{train}}$, red curve). 
	}
	\label{post_gaussian_mismatch}
\end{figure}
\FloatBarrier

\section{Potential Use of Summary Statistics Component Orthogonal to Reference LD}
\label{sec:model_null_space}

Here we explore the possibility, as mentioned in Section~\ref{lowRankApprox}, of utilizing the information contained in $\bsumStats$ that is incompatible with the reference LD but might be useful in predicting the genetic risk for the target population.

To this end, we observe that the summary statistics $\bsumStats$ can be decomposed into the projected component $\bm{P}_\textrm{ref} \bsumStats \in \colsp(\ldRef)$ and the component $\left( \bm{I} - \bm{P}_\textrm{ref} \right) \bsumStats$ orthogonal to $\colsp(\ldRef)$;
in other words, it can be expressed as
$\bsumStats = \bm{P}_\textrm{ref} \bsumStats + \left( \bm{I} - \bm{P}_\textrm{ref} \right) \bsumStats$.
And our projection approach can be viewed as modeling the information from $\bm{P}_\textrm{ref} \bsumStats$ through the approximate likelihood \eqref{likelihood_sum} as
\begin{equation}\label{eq:proj_summary_likelihood}
	\bm{P}_\textrm{ref} \bsumStats \given \bm{\beta},\bm{D}_{\text{ref}} \sim \mathcal{N}\left( \bm{D}_{\text{ref}}\bm{\beta},\frac{1}{\Ntrain} \ldRef \right),
\end{equation}
while assuming $\left( \bm{I} - \bm{P}_\textrm{ref} \right) \bsumStats$ to provide no information.

Instead of discarding the orthogonal component entirely, we can consider extracting information out of it by modeling its relation to $\bbeta$ through an additional, auxiliary likelihood.
More precisely, we assume $\bm{P}_\textrm{ref} \bsumStats$ and $\left( \bm{I} - \bm{P}_\textrm{ref} \right) \bsumStats$ to provide independent information on $\bbeta$ and assume a likelihood of the form
\begin{equation}
	\label{eq:combined_likelihood}
	L(\bsumStats \given \bbeta, \ldRef)
	= L_\mathrm{aprx} \bigr( \bm{P}_\textrm{ref} \bsumStats \given \bbeta, \ldRef) \thinnerspace
	L_\mathrm{aux} \bigl( 
	\left( \bm{I} - \bm{P}_\textrm{ref} \right) \bsumStats \given \bbeta, \ldRef
	\bigr),
\end{equation} 
where $L_\mathrm{aprx}$ denotes the approximate likelihood \eqref{eq:proj_summary_likelihood}.
In the presence of the data mismatch, there is no way to deduce a likelihood for $\left( \bm{I} - \bm{P}_\textrm{ref} \right) \bsumStats$ from the first principles as we have done for the approximate likelihood in Section~\ref{likehood_sumstat};
the auxiliary likelihood $L_\mathrm{aux}$ hence needs to be chosen heuristically and essentially is another tuning parameter.

One simple and relevant choice of $L_\mathrm{aux}$ is
\begin{equation}\label{eq:null_likelihood}
	\left(\bm{I}-\bm{P}_\textrm{ref}\right) \, \bsumStats\given\bbeta,\ldRef
	\sim\mathcal{N}\left(\sigma_\mathrm{aux}^2\bbeta, \frac{\sigma_\mathrm{aux}^2}{\Ntrain}\bm{I}\right).
\end{equation}
This choice is notable because the use of the corresponding combined likelihood \eqref{eq:combined_likelihood} has the same effect as the LD matrix regularization used by Lassosum (Section~\ref{existingmethods}).
Specifically, under the above auxiliary likelihood, the combined likelihood can be expressed as
\begin{equation}
	\label{eq:combined_likelihood_specific_form}
	\bsumStats \given \bbeta, \ldRef \sim
	\mathcal{N}\left( \left( \ldRef + \sigma_\mathrm{aux}^2 \bm{I} \right) \bbeta, \frac{1}{\Ntrain} \left( \ldRef + \sigma_\mathrm{aux}^2 \bm{I} \right) \right);
\end{equation}
i.e.\ combining the approximate likelihood for $\bm{P}_\textrm{ref} \bsumStats$ with the auxiliary likelihood for $\left( \bm{I} - \bm{P}_\textrm{ref} \right) \bsumStats$ amounts to replacing the reference LD matrix $\ldRef$ with the regularized version 
$\ldRef + \sigma_\mathrm{aux}^2\bm{I} 
\propto 
(1 - \zeta) \ldRef + \zeta \bm{I}$
for $\zeta = \frac{\sigma_\mathrm{aux}^2}{1 + \sigma_\mathrm{aux}^2}$.

Moreover, when setting $\sigma_\mathrm{aux} = 1$ and thereby inducing the LD regularization $ \ldRef + \bm{I} $, the combined likelihood \eqref{eq:combined_likelihood_specific_form} closely emulates the effect of the ad hoc solution to the data mismatch problem used by PRS-CS.
To see this, recall that the posterior of $\bbeta$ conditional on $\bsumStats$, without the projection, is given by
$$\bbeta\given \bsumStats,\ldRef,\tau, \bm{\lambda} 
\sim \mathcal{N}\left(\bm{\Phi}^{-1}\Ntrain \thinnerspace \bsumStats,\bm{\Phi}^{-1}\right),$$
where $\bm{\Phi}=\Ntrain\ldRef+\tau^{-2}\bm{\Lambda}^{-2}$.
PRS-CS's Gibbs sampler regularizes this posterior by imposing a lower bound $\tau^{-2} \lambda_j^{-2} \geq \Ntrain$ and thereby preventing the variance $\bm{\Phi}^{-1}$ from becoming too large.
In other words, PRS-CS's ad hoc fix amounts to replacing $\bm{\Phi}$ with a regularized version
$$ \bm{\Phi}_\mathrm{reg}= \Ntrain  \ldRef + \min\left\{ \Ntrain \bm{I}, \, \tau^{-2}\bm{\Lambda}^{-2} \right\}. $$
The combined likelihood \eqref{eq:combined_likelihood_specific_form} with $\sigma_\mathrm{aux} = 1$ achieves an analogous regularization since its effect amounts to replacing $\bm{\Phi}$ with
$$ \widetilde{\bm{\Phi}}_\mathrm{reg}
= \Ntrain\ldRef + \Ntrain \bm{I} + \tau^{-2} \bm{\Lambda}^{-2}, $$
where $\Ntrain \bm{I} + \tau^{-2} \bm{\Lambda}^{-2}$ provides a smoother form of regularization than the thresholding $\min\left\{ \Ntrain \bm{I}, \, \tau^{-2}\bm{\Lambda}^{-2} \right\}$.
While $\Ntrain + \tau^{-2} \lambda_j^{-2} \geq \min \left\{ \Ntrain, \, \tau^{-2} \lambda_j^{-2} \right\}$ holds in general, the two forms of regularization yield similar effects in preventing the numerical issues caused when $\tau^{-2} \lambda_j^{-2} \to 0$, in which case we have $\Ntrain + \tau^{-2} \lambda_j^{-2} \approx \min \left\{ \Ntrain, \, \tau^{-2} \lambda_j^{-2} \right\} = \Ntrain$.

In summary, the extension of our projection approach through the auxiliary modeling of $\left( \bm{I} - \bm{P}_\textrm{ref} \right) \bsumStats$ via Equation~\eqref{eq:null_likelihood} coincides with, and provides a novel interpretation of, the regularization of a reference LD matrix.
We have further shown that the ad hoc fix used by PRS-CS essentially amounts to using the regularized LD matrix $ \ldRef + \bm{I} $ or, equivalently, assuming the auxiliary likelihood with $\sigma_\mathrm{aux} = 1$ given as
\begin{equation}
	\label{eq:aux_likelihood_for_prs_cs}
	\left(\bm{I}-\bm{P}_\textrm{ref}\right) \, \bsumStats \given\bbeta,\ldRef
	\sim\mathcal{N}\left( \bbeta, \frac{1}{\Ntrain}\bm{I} \right).
\end{equation}
This perspective in particular allows us to see that PRS-CS implicitly places substantial weight on the information from training data that is incompatible with reference data.
We numerically confirm, in the supplement Section~\ref{sec:PRS-CS-proj}, that PRS-CS with regularization $ \ldRef + \bm{I} $, in place of the ad hoc constraint, closely tracks the original PRS-CS in its performances.
The empirical success of PRS-CS and its ad hoc approach suggests that it is a worthy future research direction to further explore how to best utilize the information from $\left(\bm{I}-\bm{P}_\textrm{ref}\right)\bsumStats$ or, equivalently, regularize the reference LD matrix.

\section{Auto-tuning Exponent Parameter of PRS-Bridge}\label{PRS-Bridge-auto}

Here we describe the principle behind and the implementation of the auto version of PRS-Bridge as mentioned in Section~\ref{PRSBridge} of the main manuscript.
For the auto version, we fix the percent of eigenvectors projected away for the low-rank LD approximation of Section~\ref{sec:snp_choice_and_ld_approx} as 0.8; 
our numerical studies have demonstrated this choice to yield robust performance across the traits (Section~\ref{sec:PRS-Bridge-tuning}).
This leaves the exponent parameter $\alpha$ as the only remaining tuning parameter.

Under the Bayesian framework, it is conceptually straightforward to treat $\alpha$ as an additional unknown parameter and infer it from the posterior distribution.
However, posterior inference on $\alpha$ via Gibbs sampling is known to be inefficient \citep{polson2014bayesian}.
Instead, therefore, we propose to auto-tune $\alpha$ through maximization of the marginal likelihood $L(\bsumStats \given  \ldRef, \alpha) = \int L(\bsumStats \given \ldRef, \bbeta) \, \pi(\bbeta \given \alpha) \, \mathrm{d} \bbeta$, which is also known as the model evidence \citep{mackay1992bayesian}.
This approach is equivalent to estimating the value of $\alpha$ that coincides with the mode of its marginal posterior $\pi(\alpha \mid \bsumStats, \ldRef)$ under a flat prior.

We maximize the marginal log-likelihood $\log L(\bsumStats \given  \ldRef, \alpha)$ with respect to $\alpha$ through a stochastic gradient descent algorithm as follows.
Fisher’s identity \citep{dempster1977em} tells us that
\begin{equation}\label{Fisher_identity}
	\begin{aligned}
		\nabla_{\alpha} \log L(\bsumStats \given  \ldRef, \alpha)
		&= \nabla_{\alpha} \log \pi\left( \alpha \given \bsumStats, \ldRef \right) \\
		&= \int \nabla_{\alpha} \log \pi\left( \bbeta, \tau, \alpha \given \bsumStats, \ldRef \right)
		\, \pi\left(\bbeta,\tau\given \bsumStats, \ldRef, \alpha \right)
		\, \mathrm{d} \bbeta \, \mathrm{d} \tau,
	\end{aligned}
\end{equation}
where the last expression represents the expectation of $\nabla_{\alpha} \log \pi\left( \bbeta, \tau, \alpha \given \bsumStats, \ldRef \right)$ with respect to the posterior $\pi\left(\bbeta,\tau\given \bsumStats, \ldRef \right)$.
The above identity, therefore, provides a recipe to estimate the gradient of the log marginal likelihood through Monte Carlo integration \citep{delyon1999stochastic_approx}.
The gradient $\nabla_{\alpha} \log \pi \left( \bbeta, \tau, \alpha \given \bsumStats, \ldRef \right)$ can be calculated, by noting that
$\pi \left( \bbeta, \tau, \alpha \given \bsumStats, \ldRef \right)
\propto L \left( \bsumStats \given \ldRef, \bbeta \right) \thinnerspace
\pi\left(\bbeta\given\tau, \alpha \right) \thinnerspace 
\pi(\bbeta) \thinnerspace
\pi(\tau)$, via an expression
\begin{equation}
	\label{eq:conditional_gradient}
	\begin{aligned}
		\nabla_{\alpha} \log \pi \left( \bbeta, \tau, \alpha \given \bsumStats, \ldRef \right)
		&= \nabla_{\alpha} \log	\pi \left(\bbeta\given\tau, \alpha \right) \\
		&= P/\alpha+P\frac{\Gamma'\left(1/\alpha\right)}{\Gamma\left(1/\alpha\right)}1/\alpha^2
		- \sum_{j=1}^P \, \left| \frac{\beta_j}{\tau} \right|^{\alpha} \log \, \left| \frac{\beta_j}{\tau} \right|.
	\end{aligned}
\end{equation}
In other words, we can obtain a stochastic approximation of $\nabla_{\alpha} \log L(\bsumStats \given  \ldRef, \alpha)$ by plugging in posterior samples from $\pi\left(\bbeta,\tau\given \bsumStats, \ldRef, \alpha \right)$ to the expression \eqref{eq:conditional_gradient}.
And we can use these approximate gradients to carry out the stochastic optimization.
In practice, we carry out the optimization in terms of $\log \alpha$ instead of $\alpha$;
this removes the constraint $\alpha>0$ and the gradient $\nabla_{\log \alpha} = \alpha \nabla_{\alpha}$ tends to be more well-behaved.



More explicitly, given the current value $\alpha^{(k)}$, we update the exponent parameter as follows.
We first sample from $\pi\left(\bbeta,\tau \given \bsumStats, \ldRef, \alpha^{(k)} \right)$ with the Gibbs sampler as described in Section~\ref{sec:CG_descrip}.
Using the samples $\left\{ \bbeta^{(m)}, \tau^{(m)} \right\}_{m = 1, \thinnerspace \ldots, \thinnerspace M}$, we approximate the expectation in Equation~\eqref{Fisher_identity}:
\begin{equation}
	\label{eq:stoch_grad_approx_formula}
	\widehat{\nabla}_{\log \alpha}\log \pi\left(\bsumStats\given \ldRef,\alpha^{(k)} \right)
	= P + P \frac{\Gamma'\left(1/\alpha\right)}{\Gamma\left(1/\alpha\right)}1/\alpha
	- \frac{\alpha}{M} \sum_{m=1}^M \sum_{j=1}^P  \, \left| \frac{\beta^{(m)}_j}{\tau^{(m)}} \right|^{\alpha}
	\log \, \left| \frac{\beta^{(m)}_j}{\tau^{(m)}} \right|.
\end{equation}
Finally, we update $\alpha$ via
$$\log \alpha^{(k+1)} = \log \alpha^{(k)} + \eta \widehat{\nabla}_{\log \alpha}\log \pi\left(\bsumStats\given \ldRef,\alpha^{(k)} \right),$$
where we take the step size to be $\eta = 0.001$.

For each update $\alpha^{(k)} \to \alpha^{(k+1)}$, we ran 50 iterations of the Gibbs sampler.
After the update, the Gibbs sampler requires some time to reach the new stationary distribution $\bbeta,\tau \given \bsumStats, \ldRef, \alpha^{(k + 1)}$;
we therefore discard the first 20 samples and use the remaining 30 samples in the gradient approximation \eqref{eq:stoch_grad_approx_formula}.


We note that our auto-tuning technique described above is more generally applicable to Bayesian sparse regressions based on the bridge prior, including when building PRS from individual-level data.
We provide a benchmark of PRS-Bridge-auto performances against PRS-Bridge and PRS-CS-auto in Section~\ref{sec:PRS-Bridge-auto-numerical}.

\section{Details of Conjugate Gradient-accelerated Gibbs Sampler for PRS-Bridge}
\label{sec:CG_descrip}
The bridge prior can be expressed as a scale mixture of normal distributions as \citep{polson2014bayesian}:

\begin{align*}
	\bm{\beta}\given \bm{\lambda},\tau\sim  \mathcal{N}\left(0,\tau^2\bm{\Lambda}^2\right),\text{ }
	\bm{\Lambda} = \text{diag}\left(\lambda_j\right), \
	\pi\left(\lambda_j\right)\propto \lambda_j^{-2} \pi_{st} \left(\lambda_j^{-2}/2\right),
\end{align*}
where $\tau$ is the global scale parameter, 
$\lambda_j$ the local scale parameter, 
and $\pi_{st}$ is the alpha-stable distribution with an index of stability $\alpha/2$.
The Gibbs sampler's steps are as follows:
\begin{enumerate}
	\item We first update $\tau$ by sampling from $\nu = \tau^{-\alpha}$ after marginalizing out $\lambda_j$'s. 
	Assuming a Gamma prior $\pi(\nu)\propto  \nu^{k-1}e^{-\nu\theta}$ and after integrating out $\lambda_j$'s, the posterior conditional of $\nu$ is given by
	$$\nu \given \bm{\beta}\propto \nu^{k+p/\alpha-1}\exp\left(-\nu\left(\theta+\sum_{j=1}^p |\beta_j|^{\alpha}\right)\right).$$
	\item We next update $\lambda_j$'s. 
	While the full conditional does not have a closed form formula, we can sample from it by using the double-rejection algorithm of \cite{devroye2006nonuniform} as implemented in the Python package ``bayesbridge'' in \cite{aki2022prior}.
	\item Finally, we update $\bm{\beta}$ from its full conditional given by, for $\bm{\Phi} = \Ntrain \ldRef + \tau^{-2}\bm{\Lambda}^{-2}$,
	\begin{equation}\label{post_beta}
		\begin{split}
			\bm{\beta}\given \bsumStats,\bm{\lambda},\tau&\sim \mathcal{N}\left(\bm{\Phi}^{-1}\Ntrain\bsumStats,\bm{\Phi}^{-1}\right). 
		\end{split}
	\end{equation}
\end{enumerate}
The highest computational cost arises from generating the multivariate normal distribution in~(\ref{post_beta}). 
By partitioning the whole genome into independent LD blocks, we can update $\bm{\beta}$ independently within LD block. 
To further speed up the computation, we use the conjugate gradient sampler of \cite{aki2022prior} instead of the standard approach based on Cholesky decomposition of precision matrix.
Specifically, we use a two-step procedure to sample from (\ref{post_beta}):
\begin{itemize}
	\item [1)]
	Generate a Gaussian vector $\bm{b}$ from $\mathcal{N}\left(\Ntrain\bsumStats,\bm{\Phi}\right)$ by first generating two independent Gaussian vectors, $\bm{\eta}$ from $\mathcal{N}\left(0,\bm{I}_n\right)$, and $\bm{\delta}$ from $\mathcal{N}\left(0,\bm{I}_p\right)$, and then setting
	$$\bm{b} = \Ntrain\bsumStats + \Ntrain^{1/2} \ldRef^{1/2}\bm{\eta} + \tau^{-1}\bm{\Lambda}^{-1}\bm{\delta}.$$
	
	\item [2)]
	Solve the following linear system for $\bm{\beta}$:
	\begin{equation}\label{eq4}
		\bm{\Phi}\bm{\beta} = \bm{b} \ \text{ where } \, \bm{\Phi} = \Ntrain \ldRef + \tau^{-2}\bm{\Lambda}^{-2}.
	\end{equation}
\end{itemize}
\noindent
For step (2), we use the conjugate gradient method, which is an iterative method to solve the linear system whose computational cost is dominated by matrix-vector multiplication $\bm{\Phi}\bm{\beta}$. 
To speed up the convergence of the CG method, we combine it with the prior-preconditioning strategy developed by  \cite{aki2022prior}.

\section{Details of LD Approximation Strategy}\label{sec:LD_approx}
To evaluate the impact of the choice of LD reference data on PRS methods' performance, we consider two alternative data sources for LD reference: (1) the 1000G reference samples with 489 unrelated individuals of European ancestry, (2) the UK Biobank reference samples with 337,484 unrelated individuals of European ancestry.

We additionally evaluate the impact of the choice of LD structures on the model performance. 
We provide two options to construct LD matrix in PRS-Bridge.
The first ``small-block'' option partitions the LD matrix into 1,703 independent blocks generated by the method of \cite{berisa2015approximately}, which has previously been successfully implemented in PRS-CS. 
Within chromosome 1, for example, we end up with 133 blocks with each consisting on average of 689 SNPs from the HapMap 3 SNP set.
The second ``large-block'' option partition the LD matrix using the function \texttt{snp\_ldsplit} in the R package ``bigsnpr,'' which generates the same block LD matrix as in \cite{prive2022identifying}. 
Within chromosome 1, for example, this yields 36 blocks with each consisting on average of 2,500 SNPs. 

For LDpred2, we consider both of the two LD structures described in the LDpred2 tutorial.
The ``large-block'' version uses the same large-block LD structure as described above and used by PRS-Bridge.
We use the block LD matrix estimate as provided by the LDpred2 software, which is essentially identical to the one used by PRS-Bridge except potentially for minor differences in the quality control procedure.
The ``banded'' version uses a banded LD structure with a default LD radius of 3cM.
The software does not provide a banded LD matrix estimate, so we estimate it using a random subset of 5,000 individuals from the UK Biobank reference samples.
The sub-sampling is done to ease the computational burden of estimating the banded approximation;
we experimented with the sample size of 500, 1,000, 2,000, and 5,000 and observed little improvement in the method's predictive performance once the sample size reached 2,000.

For PRS-CS and Lassosum, we consider the small-block LD structure provided as default in their software. 

\begin{figure}[!htb]
	\centering
	\includegraphics[width=.75\textwidth]{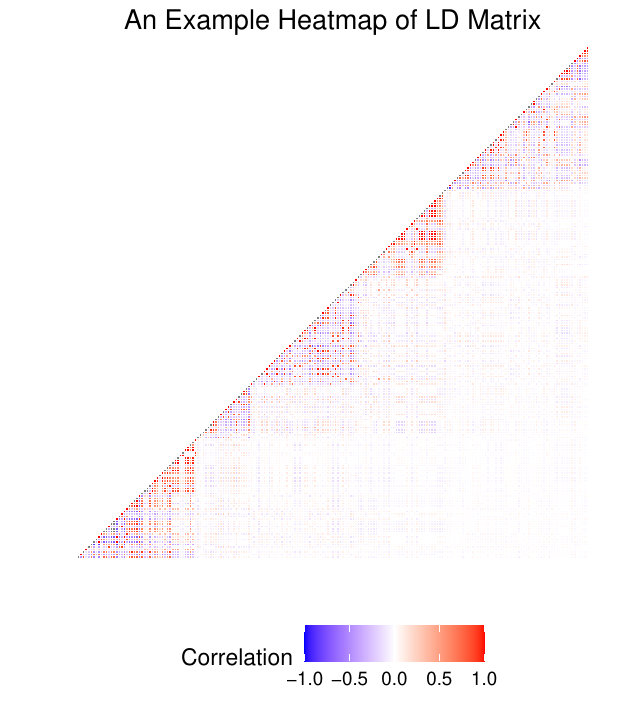}
	\caption{%
		Heatmap of the LD correlation matrix for the 200 SNPs, selected for the purpose of illustration, within an LD block (GRCh37 position: 17284065--17661178) on chromosome 22.
		The matrix is estimated from 
		unrelated European individuals from UK Biobank with complete genotype information.%
	}
	\label{ld}
\end{figure}

\section{Details of the PRS Methods' Implementations}\label{sec:method_implement}
\label{sec:method_details}
The large number of genetic variants across the genome makes for an exceptionally high-dimensional problem.
Within European populations only, there are over eight million bi-allelic SNPs with minor allele frequency above 1\%. 
To keep the computation manageable, PRS methods typically use a pre-selected subset of SNPs, the HapMap 3 SNP list providing the most commonly used one. 
This list contains approximately 1.2 million SNPs that provide good coverage of the genome for European populations. 
Following this standard practice, we include in our analysis only the HapMap 3 bi-allelic SNPs to strike a balance between statistical power and computational feasibility. 

For the implementation of LDpred2, we use the LDpred2-grid algorithm implemented in the R package ``bigsnpr'' (version 1.10.8) \citep{prive2018efficient}.
It conducts a grid search to estimate 
1) the causal SNP proportion with candidate values being a sequence of 21 numbers evenly spaced on a logarithmic scale between $10^{-5}$ and $1$, and 
2) the heritability parameter with the candidate values being the estimates from the LD score regression multiplied by 0.7, 1, or 1.4. 

For PRS-CS, we treated global shrinkage parameter $\tau$ as a tuning parameter and searched among its default candidate values, scaled by $\Ntrain^{-1}$, as suggested by the authors of PRS-CS , $\{10^{-6},10^{-4},10^{-2},1\}$. 
For our proposed PRS-Bridge, we treat the percentage of eigenvalues removed in the low-rank LD matrix approximation step as a tuning parameter with candidate values 0\%, 20\%, 40\%, 60\%, and 80\%.
The eigenvalues of magnitude less than 0.01 are always removed to prevent numerical instability.
We also treat the exponent $\alpha$ as a tuning parameter with candidate values 0.125, 0.25, and 0.5.
All the tuning parameters are selected based on the methods' performance on the tuning datasets.

For PRS-Bridge, we use a scientifically-informed prior on the global scale $\tau$ by taking advantage of the closed-form conditional variance formula
$\operatorname{var}(\beta_j\given\tau)=\frac{\Gamma\left(3/\alpha\right)}{\Gamma\left(1/\alpha\right)}\tau^2$.
Having standardized the outcome and predictors to have unit variance, we can interpret the prior conditional variance of $\beta_j$ as per-SNP heritability; i.e.\ how much each SNP on average explains the variability in the trait.
We can leverage well-established methods such as the LD-score regression to estimate the per-SNP heritability and then use this estimate to construct an informative prior on $\tau$.

\section{Comparison of Computational Speeds}\label{sec:speed}
Here we compare the computational efficiency of PRS-Bridge with LDpred2 and PRS-CS. 
We train PRS models using BMI chromosome 22 summary statistics from UK Biobank study for the three methods under various LD matrix assumptions.
We benchmark the methods' computational speed on Intel(R) Xeon(R) Silver 4310 CPU at 2.10 GHz in the single threaded setting.
The results, as summarized in Table~\ref{tab:speed}, show that PRS-Bridge (Small-block) is on average 3.7 times faster than PRS-CS (Small-block) for each tuning setting.
This allows PRS-Bridge (Small-block) to explore more tuning settings in the comparable amount of time, with an opportunity to further speed up the process by training each tuning setting in an embarrassingly parallel manner.
PRS-Bridge (Large-block) takes roughly twice the computational time of PRS-Bridge (Small-block) since updating each block of the coefficients requires sampling from multivariate normal distributions of higher dimensions. 
LDpred2 (Banded) and LDpred2 (Large-block) are faster than PRS-Bridge and PRS-CS;
this is likely due to the multiple approximations they make in their posterior computation to enforce conditional independence among the coefficients and thus avoid having to sample from multivariate normal distributions.

\FloatBarrier
\begin{table}[htb]
	\caption{Computational efficiency comparison of the three Bayesian PRS methods}
	\centering
	\label{tab:speed}
	\begin{center}
		\begin{tabular}{ |>{\centering\arraybackslash}p{5.5cm}|>{\centering\arraybackslash}p{2.7cm}|>{\centering\arraybackslash}p{2.5cm}|>{\centering\arraybackslash}p{3.5cm}| } 
			\hline
			Method & Total Time (s) & Total Tuning Settings & Average Time Per Tuning Setting (s)\\
			\hline
			PRS-Bridge (Small-block) & 1294.36 & 15 & 86.29\\ 
			PRS-Bridge (Large-block) & 2592.93 & 15 & 172.86\\ 
			PRS-CS (Small-block) & 1315.92 & 4 & 328.98\\ 
			LDpred2 (Banded) & 188.48 & 63 & 2.99\\ 
			LDpred2 (Large-block) & 24.321 & 63 & 0.39\\ 
			\hline
		\end{tabular}
	\end{center}
\end{table}

\section{Quality Control for UK Biobank and Summary Statistics Data}

For the UK Biobank data used in Section~\ref{Real_data_continuous} and \ref{Real_data_binary}, we apply the following quality control procedures.
We restrict our analysis to unrelated genotyped European participants from the UK Biobank. 
We remove SNPs with minor allele frequency (MAF) below 0.01, missing rate over 0.05, and deviation from Hardy-Weinberg equilibrium with $P<10^{-7}$. 

For the GWAS summary statistics data used in Section~\ref{Real_data_binary}, we remove SNPs 
(1) with low sample sizes ($N_j<0.9\times \text{max}(N_j)$)
(2) with outlier marginal effect size $\left(\frac{\hat{\beta}_j}{\sqrt{N}\text{SE}(\hat{\beta}_j)}\right)^2>80$; 
(3) in the long-range LD region of the position 25-35 Megabase in chromosome 6; 
(4) with MAF below 0.01. 

We use PLINK \citep{chang2015second} software for the above quality control.

\section{Results from Additional Benchmark Studies}\label{sec:add_num_result}

This section presents the results from the additional numerical studies.
We evaluate the methods' performances in the same manner as in Section~\ref{Numerical_studies}: the same metrics, datasets (for training, tuning, and validation), and uncertainty quantification based on 100 replications.

\subsection{Effect of Tuning Parameters on PRS-Bridge and PRS-CS Performances}\label{sec:PRS-Bridge-tuning}

Here we study how tuning parameter values affect PRS-Bridge and PRS-CS performances and when their performances are optimized.

We first examine what value of the exponent parameter $\alpha$ in PRS-Bridge yields the optimal performance for each trait on the tuning dataset.
Table~\ref{tab:summary_alpha} reports the optimal values identified for the six continuous and five binary traits in our benchmark study of Sections~\ref{Real_data_continuous} and~\ref{Real_data_binary}.
These optimal values provide insights into the genetic architecture of each trait, larger values suggesting greater polygenicity. 

\begin{table}[!htb]
	\tiny
	\centering
	\caption{Values of PRS-Bridge's exponent parameter found to achieve the optimal performance on the tuning dataset in the real data analyses of Sections~\ref{Real_data_continuous} and~\ref{Real_data_binary}.}
	\label{tab:summary_alpha}
	\small
	\renewcommand{\arraystretch}{1.0} 
	\setlength{\tabcolsep}{0.5pt} 
	
	\begin{tabular}{ >{\centering\arraybackslash}m{3.7cm}>{\centering\arraybackslash}m{1.5cm}>{\centering\arraybackslash}m{1.5cm}>{\centering\arraybackslash}m{5.1cm}>{\centering\arraybackslash}m{1.5cm}>{\centering\arraybackslash}m{1.5cm}}
		\hline
		Trait & \multicolumn{1}{c}{UKBB} & \multicolumn{1}{c}{1kg} & Trait & \multicolumn{1}{c}{UKBB} & \multicolumn{1}{c}{1kg}\\ 
		\hline
		BMI & 0.5 & 0.5 & Breast Cancer & 0.125 & 0.125\\
		Resting Heart Rate & 0.25 & 0.25 & Coronary Artery Disease & 0.25 & 0.5\\
		HDL & 0.25 & 0.5 & Depression & 0.25 & 0.25\\
		LDL & 0.125 & 0.25 &	Inflammatory Bowel Disease & 0.125 & 0.125\\
		APOEA & 0.25 & 0.5 & Rheumatoid Arthritis & 0.25 & 0.25\\ 
		APOEB & 0.25 & 0.25\\
		\hline
	\end{tabular}
\end{table}
\FloatBarrier

We also assess how PRS-Bridge performance depends on the proportion of eigenvectors projected away from the LD matrix.
We evaluate PRS-Bridge's performance on the tuning dataset as we vary the proportion while keeping the exponent parameter $\alpha$ fixed at the optimal value.
As shown in Figures~\ref{fig:Bridge_tuning} and~\ref{fig:Bridge_tuning_disease}, the proportion of LD eigenvectors has substantial performance impact for many of the traits, especially when the smaller 1000G data are used. 
This makes sense since the parameter controls the degree of regularization applied to the empirically estimated LD matrix, which is less stable when the reference sample size is small.
On the other hand, we also find that projecting away 80\% of eigenvectors and retaining the remaining 20\% provides a reasonable default value, yielding robust performance across the traits and providing the rationale for its use in PRS-Bridge-auto (Section~\ref{PRS-Bridge-auto}).
More generally, when using other reference data sources, the projection approach is likely most effective when the number of eigenvectors is chosen in a manner that accounts for the reference sample size.



\begin{figure}[htb]
	\centering
	\includegraphics[width=13cm]{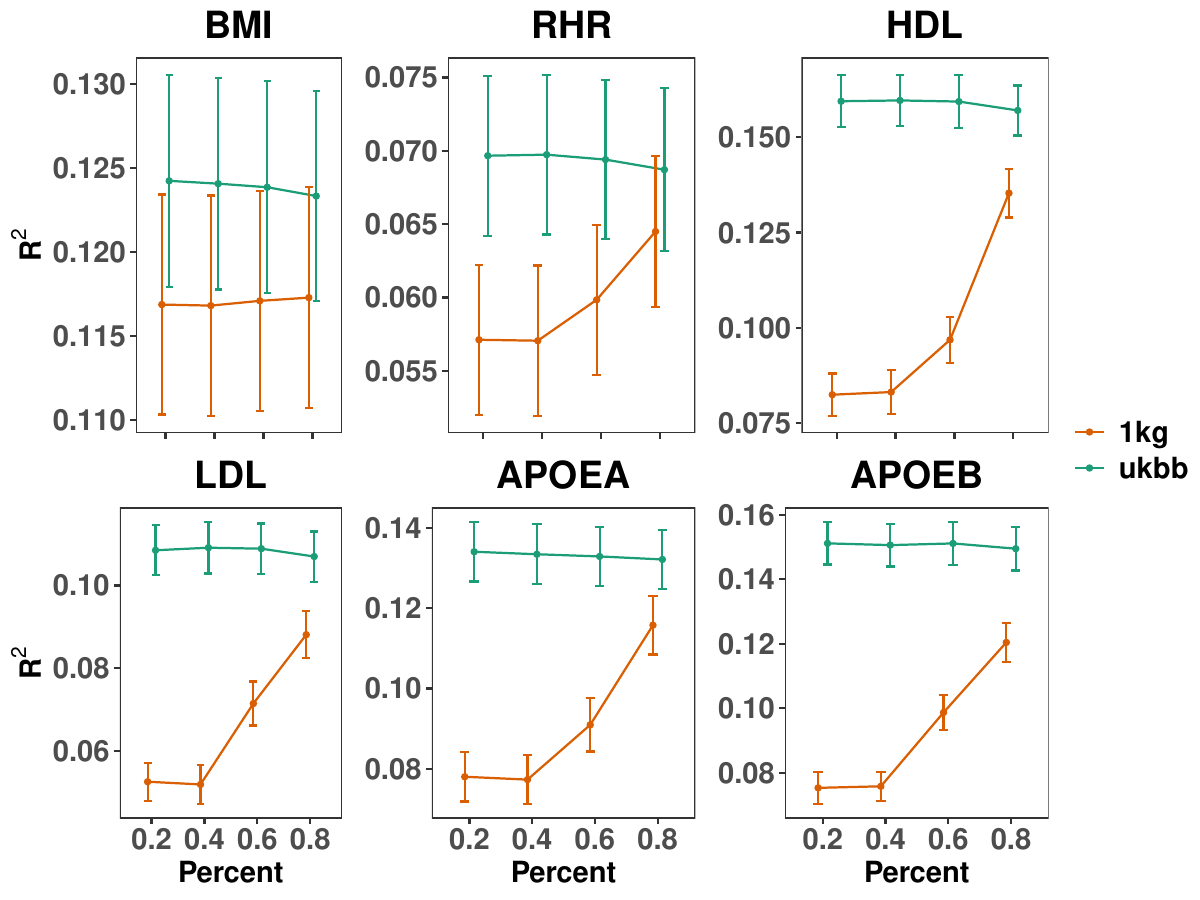}
	\caption{Out-of-sample prediction performance of PRS-Bridge as we vary the proportion of eigenvectors projected away from the LD matrix, evaluated on the six continuous traits.
		We report the average  $R^2$ as well as 1.96 times the standard error across the 100 replications. 
	}
	\label{fig:Bridge_tuning}
\end{figure}

\begin{figure}[htb]
	\centering
	\includegraphics[width=13cm]{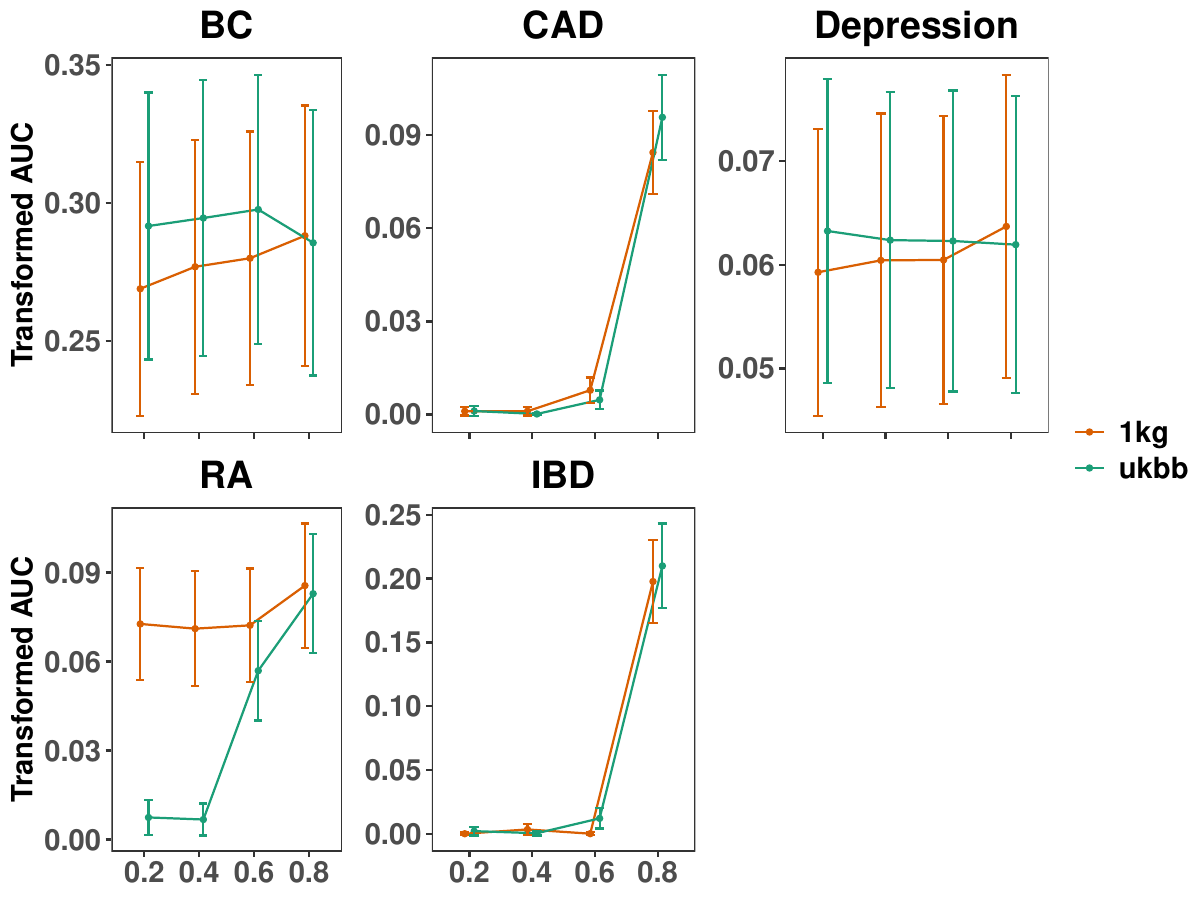}
	\caption{Out-of-sample prediction performance of PRS-Bridge as we vary the proportion of eigenvectors projected away from the LD matrix, evaluated on the five disease traits.
		We report the average transformed AUC as well as 1.96 times the standard error across the 100 replications. 
	}
	\label{fig:Bridge_tuning_disease}
\end{figure}
\FloatBarrier

We also evaluate how the hyperparameter of the PRS-CS prior affects its performance.
PRS-CS uses the Strawderman-Berger prior as its default recommended choice, but allows use of other priors from the three-parameter beta family.
This family of prior can be expressed in the global-local form as \citep{armagan2011three_param_beta}
$$\beta_j \mid \tau, \lambda_j 
\sim \mathcal{N}\left(0, \tau^2 \lambda^2_j\right)
\ \text{ with } \
\pi\left(\lambda^2_j\right)=\frac{\Gamma\left(a+b\right)}{\Gamma\left(a\right)\Gamma\left(b\right)} \left( \lambda_j^2 \right)^{a-1}\left( 1 + \lambda^2_j \right)^{-\left(a+b\right)},$$
recovering the Strawderman-Berger prior when $a = 1$ and $b = 0.5$.
Following the original work of \cite{ge2019polygenic}, we consider two alternative priors obtained under $a = 0.5$ and $a = 1.5$.
Figure~\ref{fig:PRScs_a} shows predictive performances of PRS-CS for the six continuous traits under the varying $a$ values.
We see that varying $a$ has limited impact on PRS-CS's performance, with the default recommended value $a = 1$ indeed achieving competitive performance consistently across the traits.

\begin{figure}[htb]
	\centering
	\includegraphics[width=15cm]{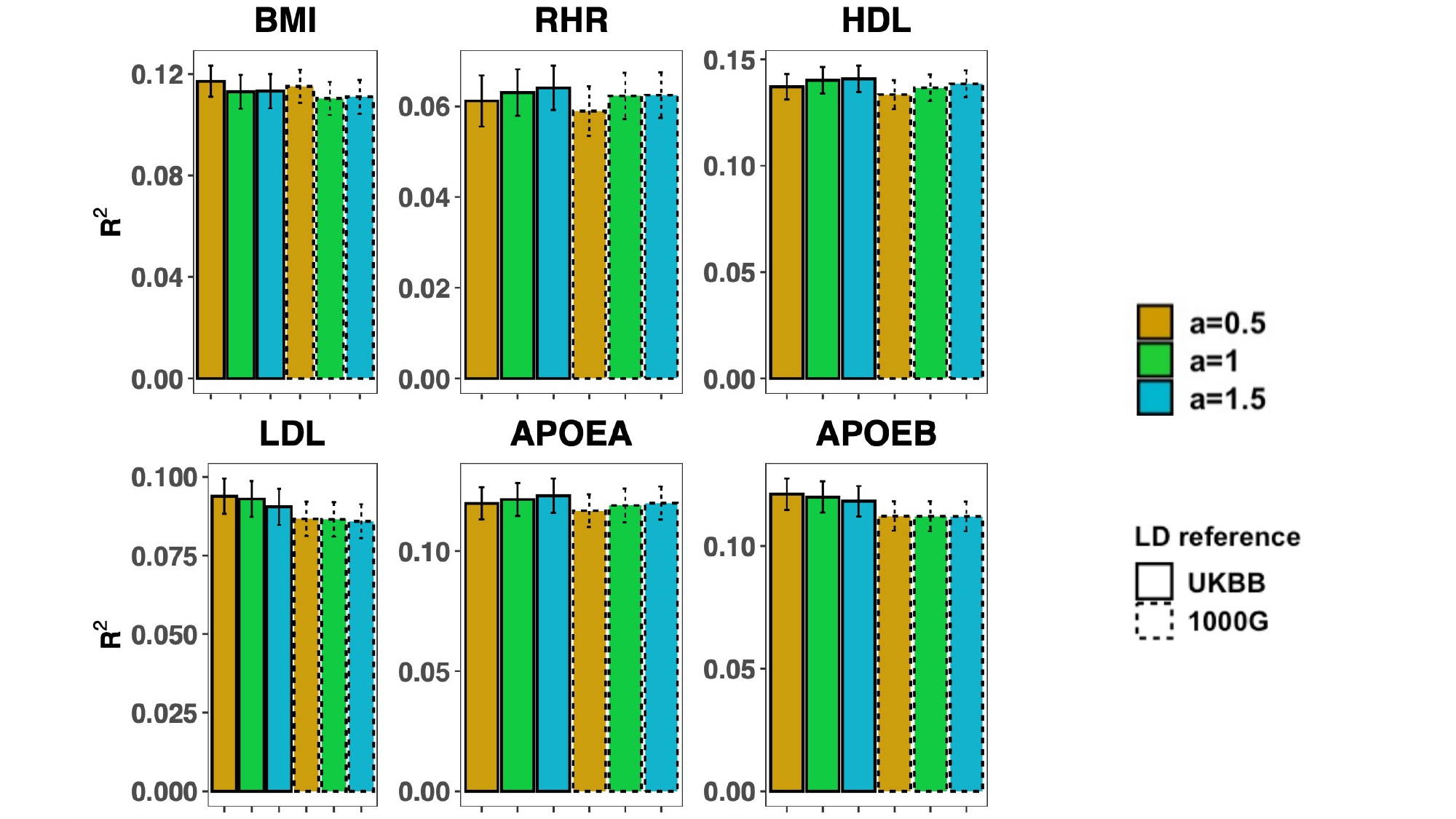}
	\caption{Out-of-sample prediction performance of PRS-CS under different values of the hyperparmeter $a$, evaluated on the six continuous traits.
	}
	\label{fig:PRScs_a}
\end{figure}
\FloatBarrier

\subsection{PRS-CS with Summary Statistics Projection}\label{sec:PRS-CS-proj-binary}

In this section, we assess the performance of PRS-CS with its ad hoc constraint on the prior variance replaced by our summary statistics projection.
This in particular allows us to compare the performances of PRS-Bridge and PRS-CS when the only difference between them is their prior choices.
In applying the projection to PRS-CS, we employ the same low-rank approximation of the LD matrix as used in PRS-Bridge and treat the percentage of eigenvectors projected away as a tuning parameter to be selected from the set of candidate values $\{0.2, 0.4, 0.6, 0.8 \}$.

Figure~\ref{fig:PRS_CS_proj_disease} presents the result of performance comparison among the original PRS-CS, PRS-CS with projection, and PRS-Bridge.
The figure shows that the projection in general does not improve PRS-CS and suggests that the observed performance advantage of PRS-Bridge is attributable to its prior choice, rather than the projection.
The projection improves PRS-CS's performance slightly for the inflammatory bowel disease, but we see that PRS-Bridge still outperforms PRS-CS with projection by a substantial margin. 

\FloatBarrier
\begin{figure}[htb]
	\centering
	\includegraphics[width=15cm]{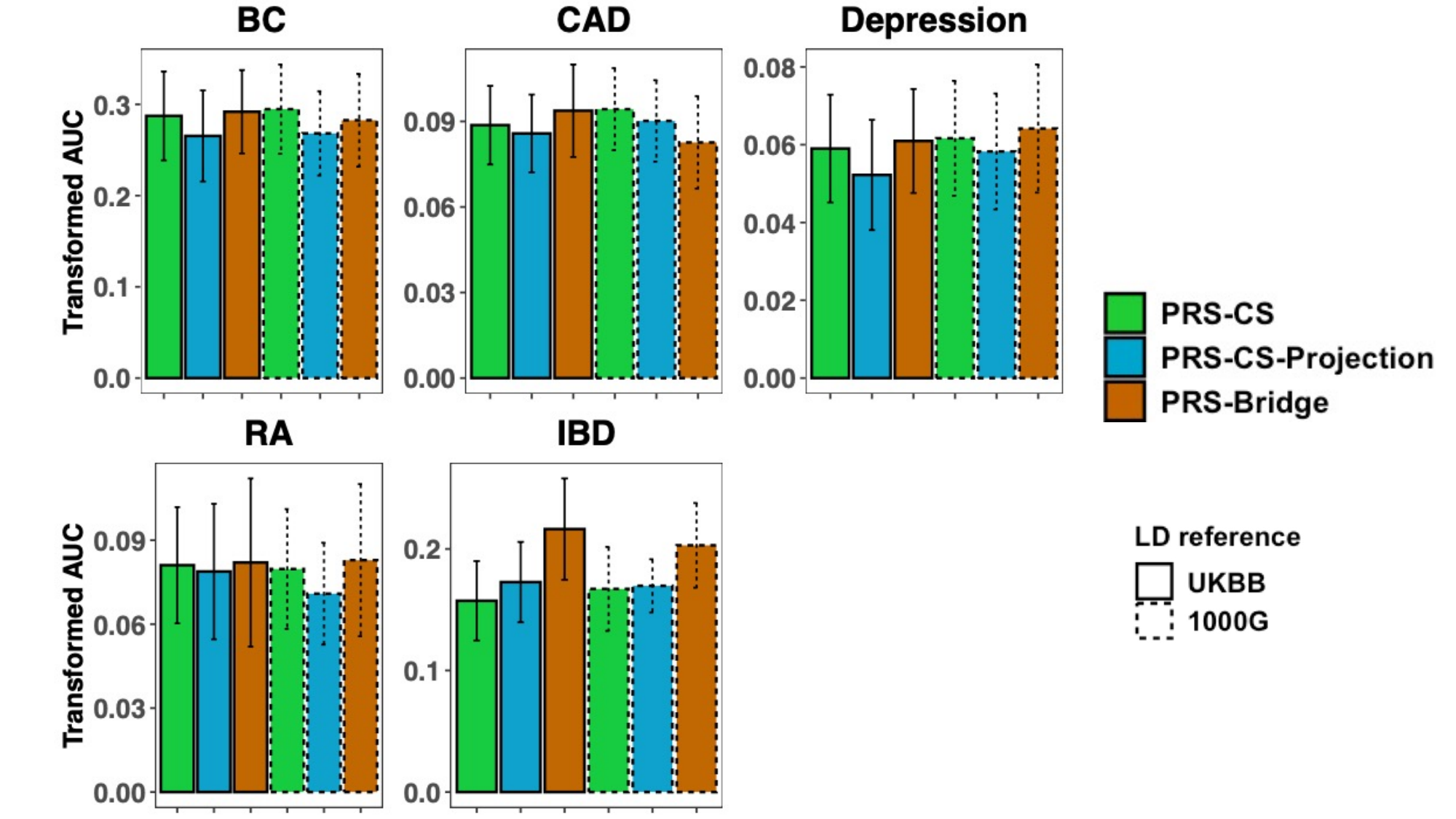}
	\caption{Out-of-sample prediction performances of PRS-CS, PRS-CS with projection, and PRS-Bridge on the five disease traits.
	}
	\label{fig:PRS_CS_proj_disease}
\end{figure}
\FloatBarrier

\subsection{PRS-CS with Alternative Prior Variance Constraints and with Regularized LD}\label{sec:PRS-CS-proj}

This section investigates the potential sensitivity, as pointed out in Section~\ref{sec:ill-behavior_demo}, of PRS-CS performance to the choice of the upper bound $\priorSdBd^2$ for the ad hoc constraint $\tau^2 \lambda_j^2 \leq \priorSdBd^2$ on the prior variance of coefficients.
We also assess the performance of PRS-CS using the regularized LD matrix $\ldRef + \bm{I}$ in place of the constraint because, as discussed in Section~\ref{sec:model_null_space}, 
this version of PRS-CS closely emulates the behavior of the default-setting PRS-CS with $\priorSdBd^2 = 1 / \Ntrain$.

Figure~\ref{PRScs_proj} presents the performance comparison among the different versions of PRS-CS: with ad hoc constraint under varying $\priorSdBd^2$ (including the default setting $\priorSdBd^2 = 1 / \Ntrain$), with projection, and with regularized LD.
The result shows the constraint-based approach to be indeed sensitive to the choice of $\priorSdBd^2$.
Among the different values of $\priorSdBd^2$, the original recommendation $\priorSdBd^2 = 1 / \Ntrain$ delivers the most consistent and competitive performance overall;
this is unsurprising in a sense---its choice as the default value by \cite{ge2019polygenic} has likely been guided by empirical evaluations on real-data PRS applications in the first place.
On the other hand, it is unclear whether this choice remains a reasonable one in more general populations or in related application areas such as transcriptomics, proteomics, and metabolomics
\citep{xu2023atlas,mosley2018probing}.

The projection-based approach achieves competitive performance, especially when the larger LD reference data from UK Biobank are used.
When the smaller reference data from 1000G are used, its performance lags behind the constraint-based approach.
This is likely because, when a reference sample size is small, the population LD structure cannot be estimated well from the reference. 
Correspondingly, the projection ends up removing, along with noise, useful information the summary statistics have about the target population.

As we have theoretically predicted in Section~\ref{sec:model_null_space}, the regularized LD approach delivers performance closely matching that of the constraint-based PRS-CS with $\priorSdBd^2 = 1 / \Ntrain$.
As also discussed in Section~\ref{sec:model_null_space}, the regularized LD approach coincides with modeling of $\left(\bm{I}-\bm{P}_\textrm{ref}\right) \, \bsumStats$ through the auxiliary likelihood \eqref{eq:aux_likelihood_for_prs_cs}, thereby recouping the information discarded by the projection.
The regularized LD approach's superior performance over the projection approach demonstrates that this component of the summary statistics indeed contains useful information on the target population, albeit being incompatible with the reference LD.
In particular, our numerical demonstration of the close correspondence between the constraint-based and regularized LD approaches  clarifies the mechanism through which the original PRS-CS manages to achieve competitive performance despite the ad hoc nature of its handling of the data mismatch problem.

\begin{figure}[!htb]
	\centering
	\includegraphics[width=15cm]{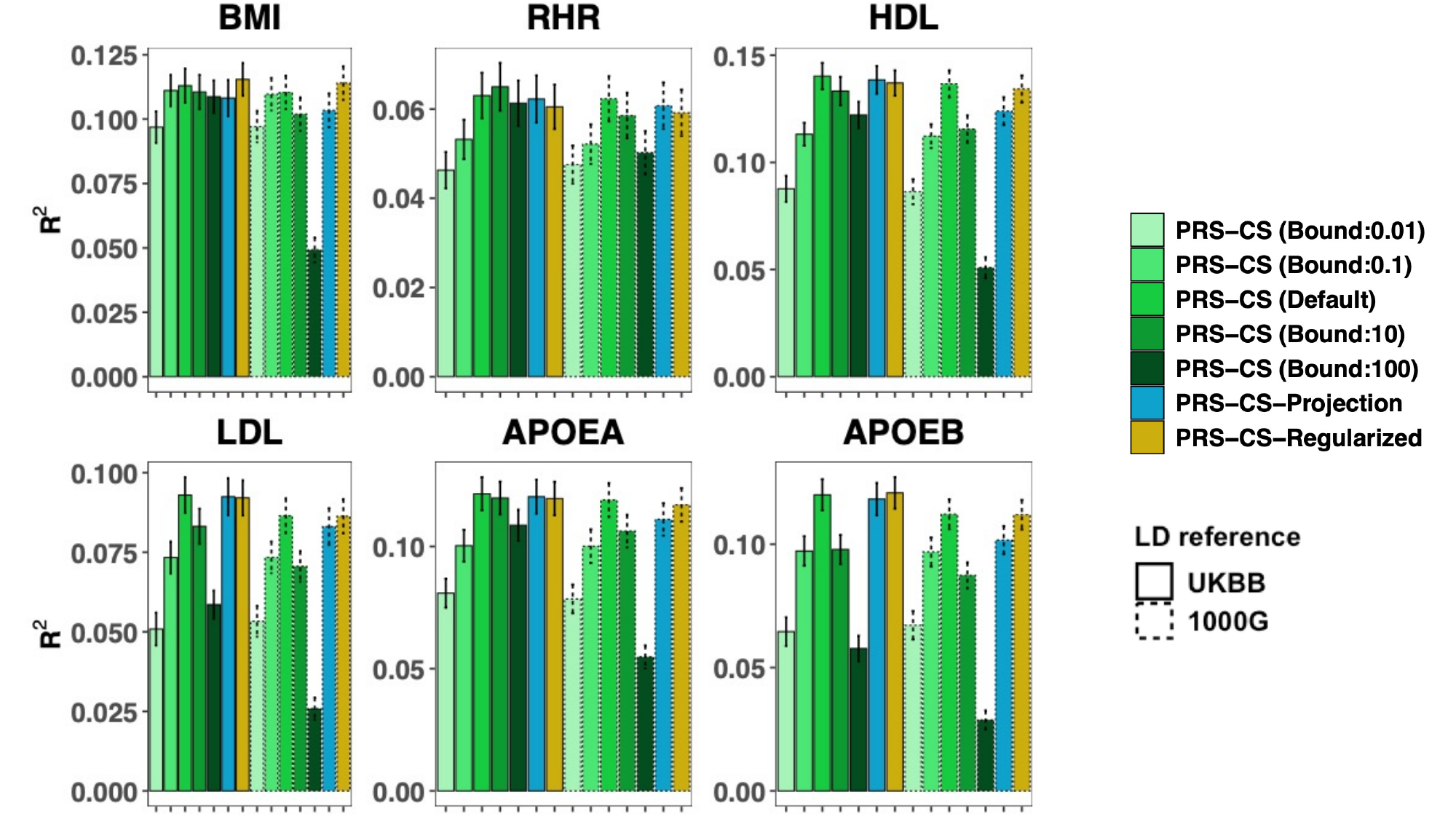}
	\caption{%
		Out-of-sample prediction performances of the different PRS-CS versions on the six continuous traits.
		In the legend, the ``Bound'' label indicates the use of the ad hoc constraint on the prior variance of coefficients, with the upper bound value $\Ntrain \priorSdBd^2$ varied among 0.01, 0.1, 1 (PRS-CS default), 10, and 100.
		The label ``Projection'' indicates the use of the summary statistics projection, and ``Regularized'' of the regularized LD matrix $\ldRef + \bm{I}$.
	}
	\label{PRScs_proj}
\end{figure}
\FloatBarrier

\subsection{Benchmarking PRS-Bridge-auto}\label{sec:PRS-Bridge-auto-numerical}

Here we assess how the auto-tuning approach described in Section~\ref{PRS-Bridge-auto} performs in practice, by benchmarking PRS-Bridge-auto against the manually tuned version of PRS-Bridge and against PRS-CS-auto.

Since PRS-Bridge-auto forgoes tuning on validation data, its performance unsurprisingly tends to lag behind the tuned PRS-Bridge, as seen in Figure~\ref{fig:UKBiobank_auto}~and~\ref{fig:Disease_auto}.
Nonetheless, PRS-Bridge-auto demonstrates strong performance on the continuous traits, with only small reductions in $R^2$ compared to PRS-Bridge.
And, compared to PRS-CS-auto, PRS-Bridge-auto achieves 9.6\% improvement in $R^2$ on average when using the UK Biobank data as reference.
When using the smaller 1000G data as reference, its performance is similar to PRS-CS-auto.
For the binary traits, PRS-Bridge-auto slightly underperforms PRS-CS-auto overall, with markedly weak performance for the coronary artery disease.
The tuned PRS-Bridge outperforms both auto methods by a substantial margin for the inflammatory bowel disease.

Overall, PRS-Bridge-auto appears to performs well when training data are sufficiently informative and hence offer a reasonable alternative when validation data are limited.
On the other hand, the manual tuning is a clearly preferred option when sizeable validation data are available.
While beyond our scope here, a related question is how well the PRS model from each method generalizes; 
i.e.\ how well it performs on populations that differ from those in the validation data source.
It could be that forgoing of validation data from a specific population gives the auto versions some advantages in terms of generalizability.
Similarly, the PRS methods with larger numbers of tuning parameters may pay some prices in terms of generalizability.

\begin{figure}[htb]
	\centering
	\includegraphics[width=15cm]{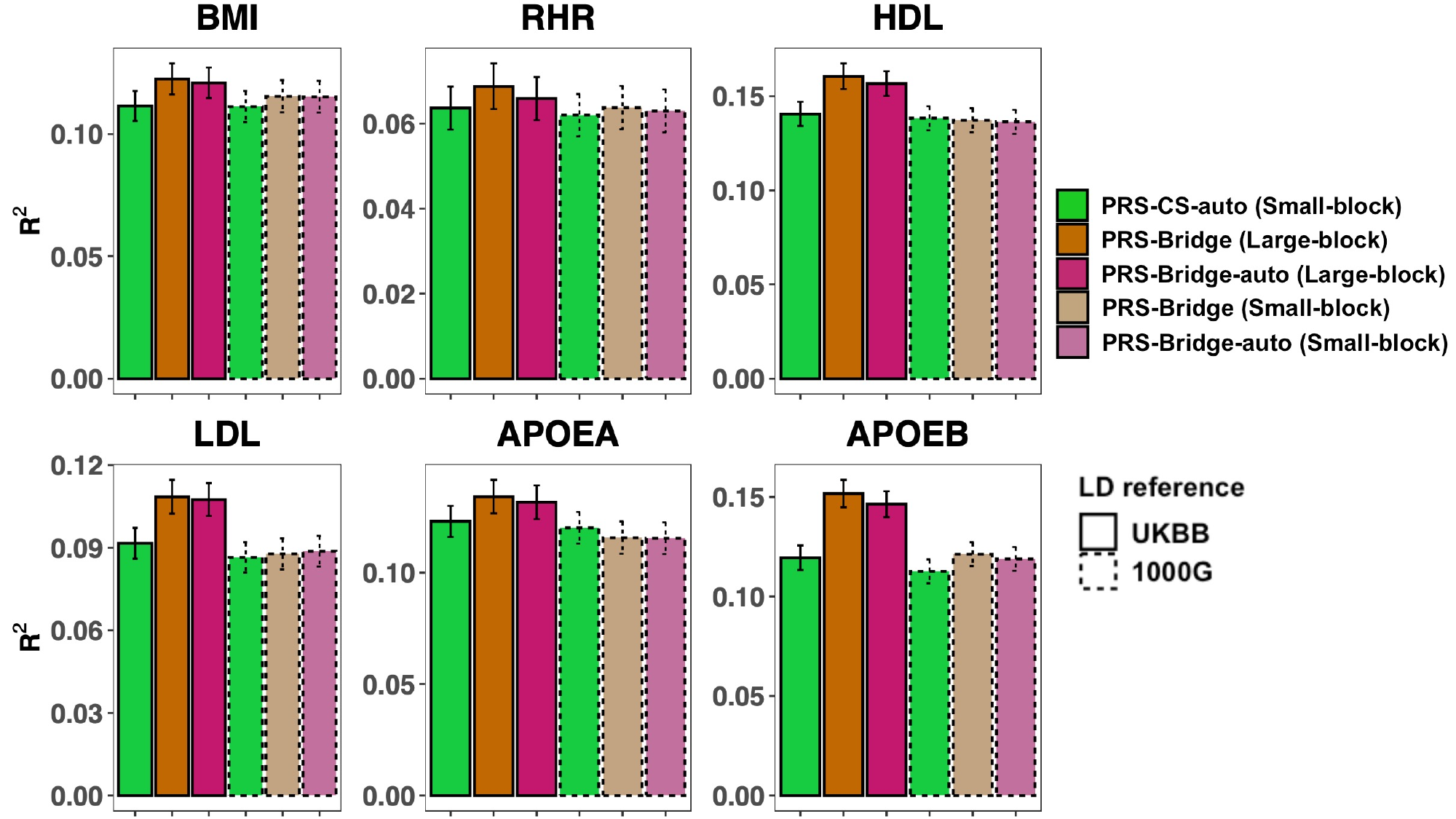}
	\caption{Out-of-sample prediction performances of PRS-CS-auto, PRS-Bridge, and PRS-Bridge-auto on the six continuous traits.
	}
	\label{fig:UKBiobank_auto}
\end{figure}

\FloatBarrier
\begin{figure}[ht]
	\centering
	\includegraphics[width=15cm]{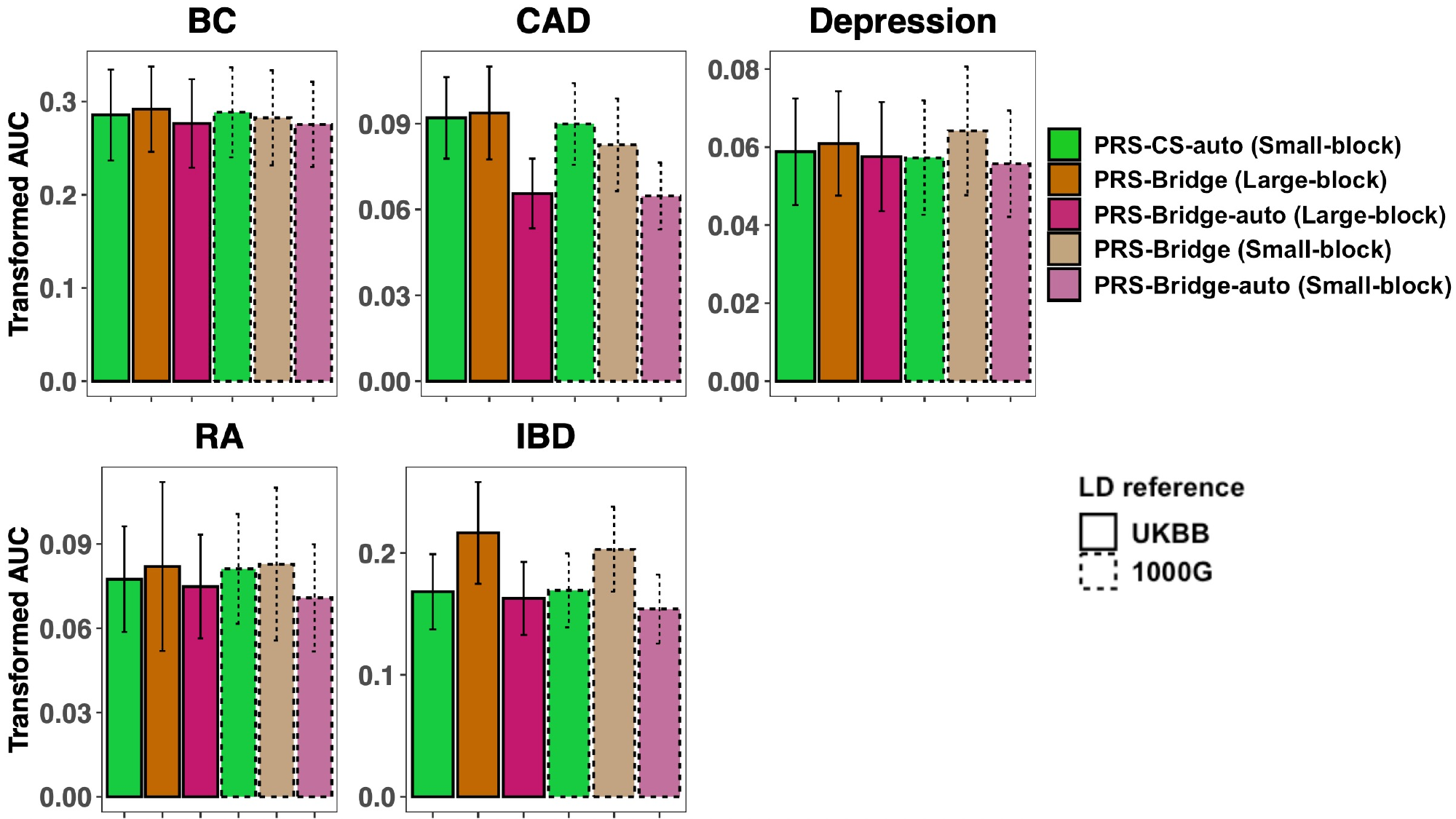}
	\caption{Out-of-sample prediction performances of PRS-CS-auto, PRS-Bridge, and PRS-Bridge-auto on the five binary disease traits.}
	\label{fig:Disease_auto}
\end{figure}
\FloatBarrier

\section{Relative Efficiency of PRS Methods}
\label{sec:relative_performance}

Here we complement the benchmark results in Section~\ref{Numerical_studies} by quantifying the PRS methods' performances and their variations in a relative sense.
Specifically, for each of the 100 replications under different validation-test data splits, we take the resulting out-of-sample $R^2$ or the transformed AUC, and calculate their ratios relative to Lassosum's.
This allows us to better assess whether one method retains an advantage over another as the test dataset varies; 
in contrast, such relative advantages can get obscured when quantifying predictive uncertainty in an absolute sense, as we have done in Figures~\ref{fig:UKBiobank}~and~\ref{fig:Disease} 

Figure~\ref{fig:UKBiobank_RE} and \ref{fig:disease_RE} below summarize our benchmark results in this relative manner, paralleling Figure~\ref{fig:UKBiobank} and~\ref{fig:Disease} summarizing the results in the absolute manner.
Here we can see that PRS-Bridge is often superior to other methods in a statistically meaningful way, even when the absolute error bars overlap in Figure~\ref{fig:UKBiobank} and~\ref{fig:Disease}.

\FloatBarrier
\begin{figure}[htb]
	\centering
	\includegraphics[width=15cm]{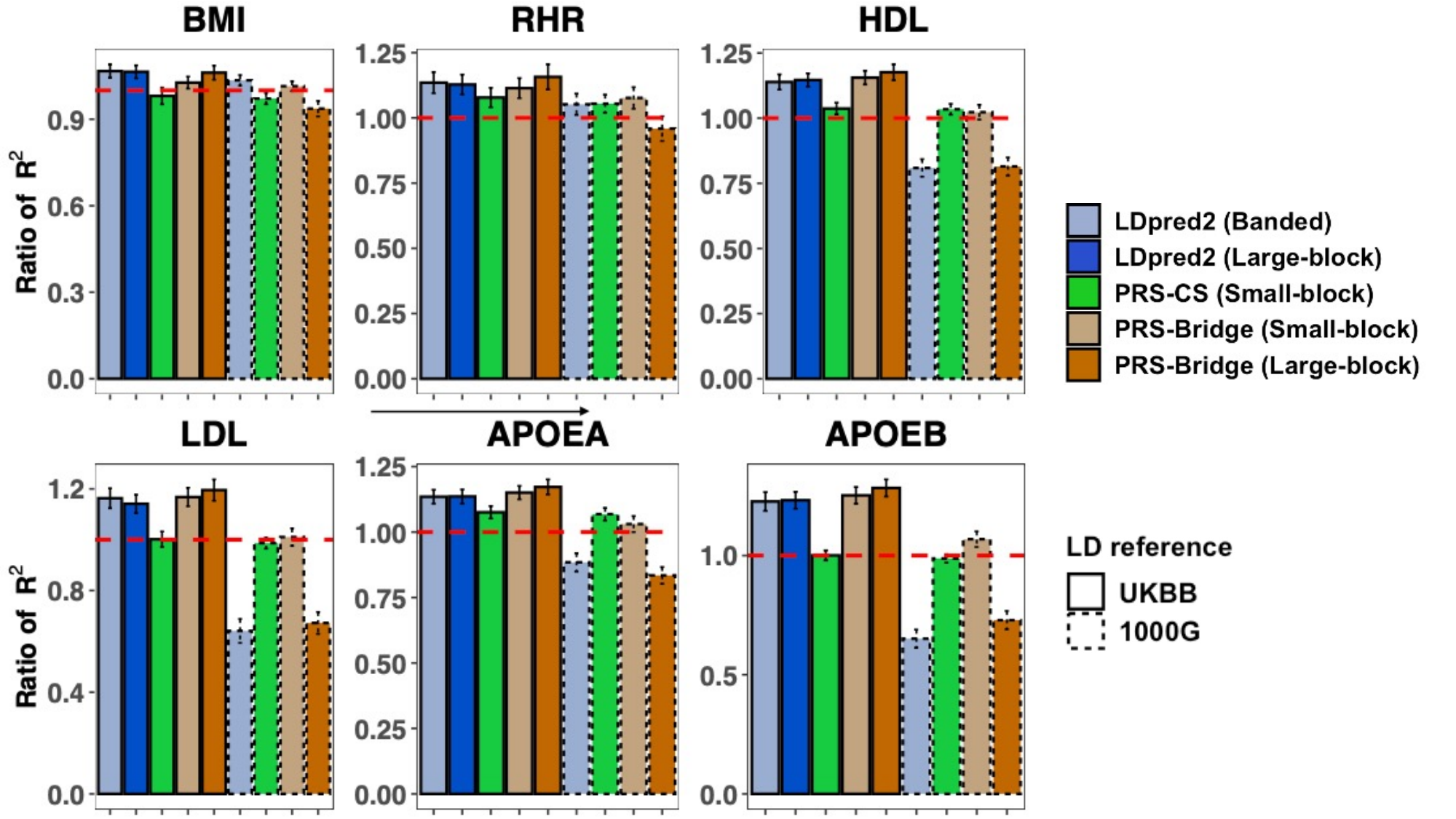}
	\caption{Ratios of out-of-sample $R^2$ of LDpred2, PRS-CS, and PRS-Bridge relative to Lassosum on the six continuous traits.
		The methods are implemented in the same manners as described in the caption of Figure~\ref{fig:UKBiobank}.
	}
	\label{fig:UKBiobank_RE}
\end{figure}
\FloatBarrier

\FloatBarrier
\begin{figure}[htb]
	\centering
	\includegraphics[width=15cm]{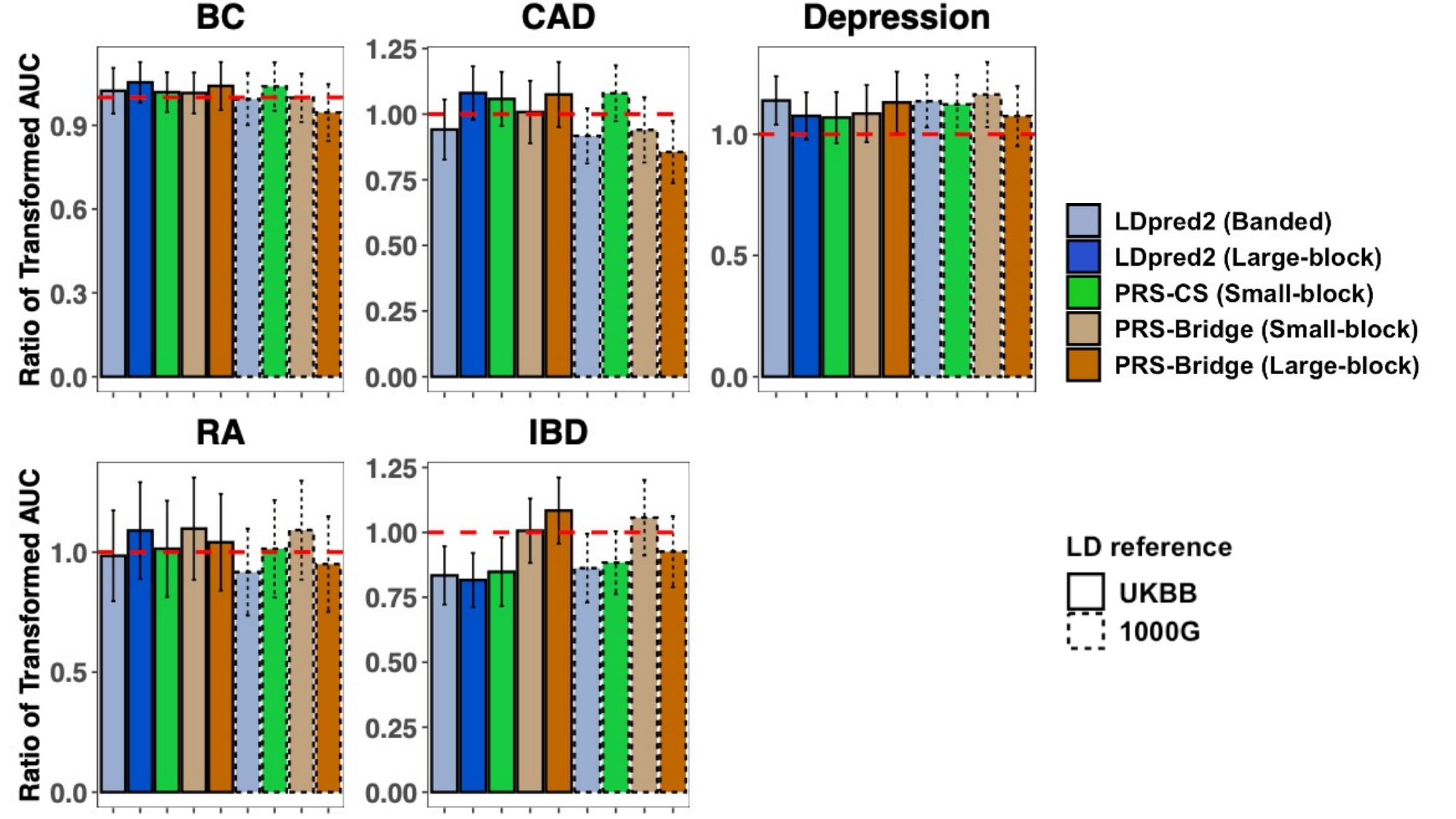}
	\caption{Ratios of out-of-sample transformed AUCs of LDpred2, PRS-CS, and PRS-Bridge relative to Lassosum on the five binary disease traits.
		The methods are implemented in the same manners as described in the caption of Figure~\ref{fig:Disease}.
	}
	\label{fig:disease_RE}
\end{figure}
\FloatBarrier

\end{document}